\documentclass[11pt]{article}
\usepackage{makeidx}
\usepackage{fullpage}
\usepackage{amsmath,amssymb,amsfonts}
\usepackage{latexsym}
\usepackage{graphicx}
\usepackage{fancybox}
\usepackage{hyperref}
\usepackage{tikz}
\usetikzlibrary{decorations.pathreplacing}
\usepackage{setspace}
\usepackage{algorithm}
\usepackage[noend]{algpseudocode}
\usepackage[framemethod=tikz]{mdframed}
\usepackage{xspace}
\usepackage{pgfplots}
\usepackage{framed}
\pgfplotsset{compat=1.5}

\newenvironment{proof}{\noindent{\bf Proof : \ }}{\hfill$\Box$\par\medskip}

\newtheorem{theorem}{Theorem}[section]
\newtheorem{corollary}[theorem]{Corollary}
\newtheorem{lemma}[theorem]{Lemma}

\newtheorem{definition}[theorem]{Definition}

\newtheorem{invariant}[theorem]{Invariant}

\newenvironment{proofof}[1]{\begin{trivlist} \item {\bf Proof
#1:~~}}
  {\qed\end{trivlist}}

\newcommand{\namedref}[2]{\hyperref[#2]{#1~\ref*{#2}}}
\newcommand{\thmlab}[1]{\label{thm:#1}}
\newcommand{\thmref}[1]{\namedref{Theorem}{thm:#1}}
\newcommand{\lemlab}[1]{\label{lem:#1}}
\newcommand{\lemref}[1]{\namedref{Lemma}{lem:#1}}

\newcommand{\corlab}[1]{\label{cor:#1}}
\newcommand{\corref}[1]{\namedref{Corollary}{cor:#1}}
\newcommand{\seclab}[1]{\label{sec:#1}}
\newcommand{\secref}[1]{\namedref{Section}{sec:#1}}

\newcommand{\figlab}[1]{\label{fig:#1}}
\newcommand{\figref}[1]{\namedref{Figure}{fig:#1}}
\newcommand{\alglab}[1]{\label{alg:#1}}
\renewcommand{\algref}[1]{\namedref{Algorithm}{alg:#1}}

\newcommand{\deflab}[1]{\label{def:#1}}
\newcommand{\defref}[1]{\namedref{Definition}{def:#1}}

\def \aadapt    {\mdef{\mathsf{A}_{adaptive}}}
\def \asingle    {\mdef{\mathsf{A}_{single}}}

\def \time    {\mdef{\mathsf{time}}}
\def \superconc    {\mdef{\mathsf{superconc}}}
\def \grates    {\mdef{\mathsf{grates}}}
\def \depth    {\mdef{\mathsf{depth}}}

\def \indeg    {\mdef{\mathsf{indeg}}}
\def \parents    {\mdef{\mathsf{parents}}}
\def \ancestors    {\mdef{\mathsf{ancestors}}}

\def \unique    {\mdef{\mathsf{UNIQUE}}}
\def \potential    {\mdef{\mathsf{PotentialParents}}}
\def \partition    {\mdef{\mathsf{BlockPartition}}}
\def \crpartition    {\mdef{\mathsf{CR-BlockPartition}}}

\def \A    {\mdef{\mathcal{A}}}
\def \d    {\mdef{\delta}}
\def \q    {\mdef{\mathbf{q}}}
\def \G    {\mdef{\mathbb{G}}}
\def \H    {\mdef{\mathcal{H}}}
\def \MHF    {\mdef{\mathsf{MHF}}}

\def \Enc    {\mdef{\mathsf{Enc}}}
\def \Dec    {\mdef{\mathsf{Dec}}}
\def \lp    {\mdef{\mathsf{lp}}}

\def \LP    {\mdef{\mathsf{LP}}}
\def \bad    {\mdef{\mathsf{bad}}}

\def \costly    {\mdef{\mathsf{COSTLY}}}
\def \setup    {\mdef{\textbf{setup}}}
\def \challenge    {\mdef{\textbf{challenge}}}

\def \guess    {\mdef{\textbf{guess}}}
\def \mhfeval  {\mdef{\mathsf{MHF.Eval}}}

\def \valiant    {\mdef{\mathsf{Valiant}}}
\def \request    {\mdef{\mathsf{request}}}
\def \store    {\mdef{\mathsf{store}}}
\def \load    {\mdef{\mathsf{load}}}
\def \True    {\mdef{\mathsf{True}}}
\def \False    {\mdef{\mathsf{False}}}
\def \invariant    {\mdef{\mathsf{invariant}}}

\def \Write    {\mdef{\mathsf{Write}}}
\def \Read    {\mdef{\mathsf{Read}}}
\def \overlay    {\mdef{\mathsf{overlay}}}
\def \SC    {\mdef{\mathsf{SC}}}

\newcommand{\PPr}[1]{\ensuremath{\mathbf{Pr}\left[#1\right]}}
\newcommand{\PPPr}[2]{\ensuremath{\underset{#1}{\mathbf{Pr}}\left[#2\right]}}

\newcommand{\EEx}[2]{\ensuremath{\underset{#1}{\mathbb{E}}\left[#2\right]}}
\renewcommand{\O}[1]{\ensuremath{\mathcal{O}\left(#1\right)}}
\newcommand{\eps}{\epsilon}

\newcommand{\mdef}[1]{{\ensuremath{#1}}\xspace}  
\newcommand{\myfunc}[1]{\mdef{\mathsf{#1}}}      

\DeclareMathOperator*{\poly}{poly}

\newcommand{\superscript}[1]{\ensuremath{^{\mbox{\tiny{\textit{#1}}}}}\xspace}
\def \th {\superscript{th}}     
\def \etal{{\it et~al.}}

\def \negl     {\mdef{\myfunc{negl}}}                

\newcommand{\ignore}[1]{}

\def \cc       {\mdef{\mathsf{cc}}}
\newcommand{\Peb}{{\cal P}} 

\newcommand{\pPeb}{\Peb^{\parallel}}

\def \pcc {\cc} 

\newif\ifnotes\notestrue 
\ifnotes
\newcommand{\samson}[1]{\textcolor{purple}{{\bf (Samson:} {#1}{\bf ) }} \marginpar{\tiny\bf
             \begin{minipage}[t]{0.5in}
               \raggedright S:
            \end{minipage}}}       
            
\newcommand{\jeremiah}[1]{\textcolor{red}{{\bf (Jeremiah:} {#1}{\bf ) }} \marginpar{\tiny\bf
             \begin{minipage}[t]{0.5in}
               \raggedright J:
            \end{minipage}}    }

\else
\newcommand{\samson}[1]{}
\newcommand{\jeremiah}[1]{}
\fi

\hypersetup{
     colorlinks   = true,
     citecolor    = blue,
		 linkcolor		= red
}

\makeatletter
\renewcommand*{\@fnsymbol}[1]{\textcolor{mahogany}{\ensuremath{\ifcase#1\or *\or \dagger\or \ddagger\or
 \mathsection\or \triangledown\or \mathparagraph\or \|\or **\or \dagger\dagger
   \or \ddagger\ddagger \else\@ctrerr\fi}}}
\makeatother

\providecommand{\email}[1]{\href{mailto:#1}{\nolinkurl{#1}\xspace}}

\definecolor{electricpurple}{rgb}{0.75, 0.0, 1.0}
\definecolor{fluorescentpink}{rgb}{1.0, 0.08, 0.58}
\definecolor{mahogany}{rgb}{0.75, 0.25, 0.0}
\definecolor{darkblue}{rgb}{0.0, 0.0, 0.55}
\definecolor{darkpastelgreen}{rgb}{0.01, 0.75, 0.24}
\definecolor{darkgreen}{rgb}{0.0, 0.2, 0.13}
\definecolor{darkgoldenrod}{rgb}{0.72, 0.53, 0.04}
\definecolor{darkred}{rgb}{0.55, 0.0, 0.0}
\hypersetup{
     colorlinks   = true,
     citecolor    = electricpurple,
		 linkcolor		= blue
}

\title{Computationally Data-Independent Memory Hard Functions}
\author{Mohammad Hassan Ameri\thanks{Department of Computer Science, Purdue University. 
E-mail: \email{mameriek@purdue.edu}}
\and
Jeremiah Blocki\thanks{Department of Computer Science, Purdue University. 
E-mail: \email{jblocki@purdue.edu}}
\and
Samson Zhou\thanks{School of Computer Science, Carnegie Mellon University. 
E-mail: \email{samsonzhou@gmail.com}}}
\date{\today}

\begin{document}
\maketitle

\begin{abstract}
Memory hard functions (MHFs) are an important cryptographic primitive that are used to design egalitarian proofs of work and in the construction of moderately expensive key-derivation functions resistant to brute-force attacks. 
Broadly speaking, MHFs can be divided into two categories: data-dependent memory hard functions (dMHFs) and data-independent memory hard functions (iMHFs). 
iMHFs are resistant to certain side-channel attacks as the memory access pattern induced by the honest evaluation algorithm is independent of the potentially sensitive input e.g., password. 
While dMHFs are potentially vulnerable to side-channel attacks (the induced memory access pattern might leak useful information to a brute-force attacker), they can achieve higher cumulative memory complexity (CMC) in comparison than an iMHF. 
In particular, any iMHF that can be evaluated in $N$ steps on a sequential machine has CMC {\em at most} $\O{\frac{N^2\log\log N}{\log N}}$. 
By contrast, the dMHF scrypt achieves maximal CMC $\Omega(N^2)$ --- though the CMC of scrypt would be reduced to just $\O{N}$ after a side-channel attack. 

In this paper, we introduce the notion of computationally data-independent memory hard functions (ciMHFs). 
Intuitively, we require that memory access pattern induced by the (randomized) ciMHF evaluation algorithm appears to be independent from the standpoint of a computationally bounded eavesdropping attacker --- even if the attacker selects the initial input. 
We then ask whether it is possible to circumvent known upper bound for iMHFs and build a ciMHF with CMC $\Omega(N^2)$. 
Surprisingly, we answer the question in the affirmative when the ciMHF evaluation algorithm is executed on a two-tiered memory architecture (RAM/Cache). 

We introduce the notion of a $k$-restricted dynamic graph to quantify the continuum between unrestricted dMHFs $(k=n)$ and iMHFs ($k=1$). 
For any $\eps > 0$ we show how to construct a $k$-restricted dynamic graph with $k=\Omega(N^{1-\eps})$ that provably achieves maximum cumulative pebbling cost $\Omega(N^2)$. 
We can use $k$-restricted dynamic graphs to build a ciMHF provided that cache is large enough to hold $k$ hash outputs and the dynamic graph satisfies a certain property that we call ``amenable to shuffling''.  
In particular, we prove that the induced memory access pattern is indistinguishable to a polynomial time attacker who can monitor the locations of read/write requests to RAM, but not cache. 
We also show that when $k=o\left(N^{1/\log\log N}\right)$, then any $k$-restricted graph with constant indegree has cumulative pebbling cost $o(N^2)$. 
Our results almost completely characterize the spectrum of $k$-restricted dynamic graphs. 
\end{abstract}

\section{Introduction}
Memory hard functions (MHFs)~\cite{AbadiBMW05, Percival09} are a central component in the design of password hashing functions~\cite{BiryukovDK15}, egalitarian proofs of work~\cite{DworkN92}, and moderately hard key-derivation functions~\cite{Percival09}. 
In the setting of password hashing, the objective is to design a function that can be computed relatively quickly on standard hardware for honest users, but is prohibitively expensive for an offline attacker to compute millions or billions of times while checking each password in a large cracking dictionary. 
The first property allows legitimate users to authenticate in a reasonable amount of time, while the latter goal discourages brute-force offline guessing attacks, even on low-entropy secrets such as passwords, PINs, and biometrics. 
The objective is complicated by attackers that use specialized hardware such as Field Programmable Gate Arrays (FPGAs) or Application Specific Integrated Circuits (ASICs) to significantly decrease computation costs by several orders of magnitude, compared to an honest user using standard hardware. 
For example, the Antminer S17, an ASIC Bitcoin miner exclusively configured for SHA256 hashes, can compute up to 56 trillion hashes per second, while the rate of many standard CPUs and GPUs are limited to 200 million hashes per second and 1 billion hashes per second, respectively. 

Memory hard functions were developed on the observation that memory costs such as chip area tend to be equitable across different architectures. 
Therefore, the cost of evaluating an ideal ``egalitarian'' function would be dominated by memory costs. 
Blocki \etal~\cite{BlockiHZ18} argued that key derivation functions without some form of memory hardness provide insufficient defense against a economically motivated offline attacker under the constraint of reasonable authentication times for honest users. 
In fact, most finalists in the 2015 Password Hashing Competition claimed some form of memory hardness~\cite{ForlerLW14, BiryukovDK15, SimplicioAASB15}. 
To quantify these memory costs, memory hardness~\cite{Percival09} considers the cost to build, obtain, and empower the necessary hardware to compute the function. 
One particular metric heavily considered by recent cryptanalysis~\cite{AlwenB16,AlwenBP17,AlwenBH17,AlwenB17,BlockiZ17} is cumulative memory complexity (CMC)~\cite{AlwenS15}, which measures the amortized cost of any parallel algorithm evaluating the function on several distinct inputs. 
Despite known hardness results for quantifying~\cite{BlockiZ18} or even approximating~\cite{BlockiLZ19} a function's CMC, even acquiring asymptotic bounds provide automatic bounds for other attractive metrics such as space-time complexity~\cite{LengauerT82} or energy complexity~\cite{RenD17,BlockiRZ18}.

\paragraph{Data-Dependent vs. Data-Independent Memory Hard Functions.}
At a high level, memory hard functions can be categorized into two design paradigms: data-dependent memory hard functions (dMHFs) and data-independent memory hard functions (iMHFs). 
dMHFs induce memory access patterns that depend on the input, but can achieve high memory hardness with potentially relatively easy constructions~\cite{AlwenCPRT17}. 
However, dMHFs are also vulnerable to side-channel attacks due to their inherent data dependent memory access patterns~\cite{Bernstein05}. 
Examples of dMHFs include scrypt~\cite{Percival09}, Argon2d~\cite{BiryukovDK16} and Boyen's halting puzzles~\cite{Boyen07}. 
On the other hand, iMHFs have memory access patterns that are independent of the input, and therefore resist certain side-channel attacks such as cache timing~\cite{Bernstein05}. 
Examples of iMHFs include 2015 Password Hashing Competition (PHC) winner Argon2i~\cite{BiryukovDK15}, Balloon Hashing~\cite{BonehCS16} and DRSample~\cite{AlwenBH17}. 
iMHFs with high memory hardness can be more technically challenging to design, but even more concerning is the inability of iMHFs to be maximally memory hard. 

Alwen and Blocki~\cite{AlwenB16} proved that the CMC of any iMHF running in time $N$ is at most $\O{\frac{N^2\log\log N}{\log N}}$, while the dMHF scrypt has cumulative memory complexity $\Omega(N^2)$~\cite{AlwenCPRT17}, which matches the maximal amount and cannot be obtained by any iMHF. 
However, the cumulative memory complexity of a dMHF can be greatly decreased through a side-channel attack, if an attacker has learned the memory access pattern induced by the true input. 
Namely, a brute-force attacker can preemptively quit evaluation on a guess $y$ once it is clear that the induced memory access pattern on input $y$ differs from that on the true input $x$. For example, the cumulative memory complexity of scrypt after a side-channel attack is just $\O{N}$. 

Ideally, we would like to obtain a family of memory hard functions with cumulative memory complexity $\Omega(N^2)$ without any vulnerability to side-channel attacks. 
A natural approach would be some sort of hybrid between data-dependent and data-independent modes, such as Argon2id, which runs the MHF in data-independent mode for $\frac{N}{2}$ steps before switching to data-dependent mode for the final $\frac{N}{2}$ steps. 
Although the cumulative memory complexity is the maximal $\Omega(N^2)$ if there is no side-channel attack, the security still reduces to that of the underlying iMHF (e.g., Argon2i) if there is a side-channel attack. Hence even for a hybrid mode, the cumulative memory complexity is just $\O{\frac{N^2\log\log N}{\log N}}$ (or lower) in the face of a side-channel attack. 
Thus in this paper we ask:
\begin{quote}
In the presence of side-channel attacks, does there exist a family of functions with $\Omega(N^2)$ cumulative memory complexity?
\end{quote}

\subsection{Our Contributions}
Surprisingly, we answer the above question in the affirmative for a natural class of side-channel attacks that observe the read/write memory locations. 
We introduce the concept of computationally data-independent memory hard functions to overcome the inability of data-independent memory hard functions to be maximally memory hard~\cite{AlwenB16} without the common side-channel vulnerabilities of data-dependent memory hard functions~\cite{Bernstein05}. 
Our constructions work by randomly ``shuffling'' memory blocks in cache before they are stored in RAM (where the attacker can observe the locations of read/write requests). 
Intuitively, each time $\mhfeval(x)$ is executed the induced memory access pattern will appear different due to this scrambling step. The goal is to ensure that an attacker can not even distinguish between the observed memory access pattern on two known inputs $x \neq y$. 

Towards this goal we define $k$-restricted dynamic graphs as a tool to quantify the continuum between dMHFs and iMHFs. Intuitively, in a $k$-restricted dynamic graph $G = (V=[N],E)$ we have $\parents(v) = \{v-1, r(v)\}$ where the second (data-dependent) parent $r(v) \in R_v$ must be selected from a fixed (data-independent) restricted set $R_v \subseteq V$ of size $|R_v| \leq k$. When $k=1$ the function is an iMHF (the parent $r(v) \in R_v$ of each node $v$ is fixed in a data-independent manner) and when $k=N$ the function is an unrestricted dMHF --- scrypt and Argon2d are both examples of unrestricted dMHFs. Intuitively, when $k$ is small it becomes easier to scramble the labels $R_v$ in memory so that the observed memory access patterns on two known inputs $x \neq y$ are computationally indistinguishable.

We then develop a graph gadget that generates a family of ciMHFs using $k$-restricted graphs. Using this family of ciMHFs, we characterize the tradeoffs between the value of $k$ and the overall cumulative memory cost of $k$-restricted graphs.

\paragraph{Impossibility Results for Small $k$.}
Since $k$-restricted graphs correspond to iMHFs for $k=1$, and it is known that $\cc(\G)=\O{\frac{N^2\log\log N}{\log N}}$ for any family $\G$ of iMHFs~\cite{AlwenB16}, then one might expect that it is impossible to obtain maximally memory hard ciMHFs for small $k$. 
Indeed, our first result shows that this intuition is correct; we show that for any $k=o\left(N^{1/\log\log N}\right)$, then any family of $k$-restricted graphs $\G$ with constant indegree has $\cc(\G)=o(N^2)$. 
\begin{theorem}
\thmlab{thm:informal:lower}
Let $\G$ be any family of $k$-restricted dynamic graphs with constant $\indeg(\G)$. 
Then 
\[\cc(\G)=\O{\frac{N^2}{\log\log N}+N^{2-1/2\log\log N}\sqrt{k^{1-1/\log\log N}}}.\]
Thus for $k=o\left(N^{1/\log\log N}\right)$, we have $\cc(\G)=o(N^2)$.
\end{theorem}
We prove this result in \thmref{thm:attack} and \corref{cor:attack} in \secref{sec:generic} by generalizing ideas from the pebbling attack of Alwen and Blocki~\cite{AlwenB16} against any iMHF to $k$-restricted dynamic graphs graphs. 
The pebbling attack of Alwen and Blocki~\cite{AlwenB16} exploited the fact that any constant indegree DAG $G$ is somewhat depth-reducible e.g., we can always find a set $S \subseteq V(G)$ of size $e=\O{\frac{N \log \log N}{\log N}}$ such that any path in $G-S$ has length at most $d=\frac{N}{\log^2 N}$. 
The attack then proceeds in a number of \emph{light phases} and \emph{balloon phases}, where the goal of light phase $i$ is to place pebbles on the interval $[ig+1,(i+1)g]$, for some parameter $g$ to be optimized. 
At the same time, the attacker discards pebbles on all nodes $v$ unless $v \in S$ or unless $v$ is a parent of one of the next $g$ nodes $[ig+1,(i+1)g]$ that we want to pebble. 
Once light phase $i$ is completed, balloon phase $i$ uses the pebbles on $S$ to recover all previously discarded pebbles. 
Note that balloon phase $i$ thus promises that pebbles are placed on the parents of the nodes $[(i+1)g+1,(i+2)g]$, so that light phase $i+1$ can then be initiated and so forth.

One key difference is that we must maintain pebbles on all $gk$ nodes $u \in \bigcup_{v \in [ig+1,(i+1)g]} R_v$ that are ``potential parents'' of the next $g$ nodes $[ig+1,(i+1)g]$. 
The total cost of the pebbling attack is $\O{eN+gkN + \frac{N^2d}{g}}$, which is identical to \cite{AlwenB16} when $k=1$ for $(e,d)$-reducible DAGs. 
In general for small values of $k$, the dynamic pebbling strategy can still achieve cumulative memory cost $o(N^2)$ after optimizing for $g$. 

\paragraph{Maximally Hard $k$-restricted dMHF.}
In \secref{sec:construct}, we show how to construct a $k$-restricted dynamic graph for $k=\O{N^{1-\eps}}$, which has cumulative pebbling cost $\Omega(N^2)$ for any constant $\eps > 0$. 
Intuitively, our goal is to force the pebbling strategy to maintain $\Omega(N)$ pebbles on the graph for $\Omega(N)$ steps or pay a steep penalty. 
In particular, we want to ensure that if there are $o(N)$ pebbles on the graph at time $i$ then the cumulative pebbling cost to advance a pebble just $2k = \O{N^{\eps}}$ steps is at least $\Omega(N^{2-\eps})$ with high probability. 
This would imply that the pebbling strategy either keeps $\Omega(N)$ pebbles on the graph for $\Omega(N)$ steps or that the pebbling strategy pays a penalty of $\Omega(N^{2-\eps})$ at least $\Omega\left(\frac{N}{k}\right) = \Omega(N^{\eps})$ times. 
In either case the cumulative pebbling cost will be $\Omega(N^2)$.

One of our building blocks is the ``grates'' construction of Schnitger~\cite{Schnitger83} who, for any $\eps>0$, showed how to construct a constant indegree DAG $G_\eps$ that is $(e,d)$-depth robust graph with $e=\Omega(N)$ and $d=\Omega(N^{1-\eps})$. 
Our second building block is the superconcentrator~\cite{Pippenger77, LengauerT82} graph. 
By overlaying the DAG $G_{\eps}$ with a superconcentrator, we can spread out the data-dependent edges on the top layer of our graph to ensure that (with high probability) advancing a pebble $2k = \O{N^{\eps}}$ steps on the top layer starting from a pebbling configuration with $o(N)$ pebbles on the graph requires us to repebble an $(e,d)$-depth robust graph with $e=\Omega(N)$ and $d=\Omega(N^{1-\eps})$. 
This is sufficient since Alwen \etal~\cite{AlwenBP17} showed that the cumulative pebbling cost of any $(e,d)$-depth robust graph is at least $ed$.

\paragraph{Open Question:} 
We emphasize that we only show that any dynamic pebbling strategy for our $k$-restricted dynamic graph has cumulative cost $\Omega(N^2)$. 
This is not quite the same as showing that our dMHF has CMC $\Omega(N^2)$ in the parallel random oracle model. 
For static graphs, we know that the CMC of an iMHF is captured by the cumulative pebbling cost of the underlying DAG~\cite{AlwenS15}. 
We take the dynamic pebbling lower bound as compelling evidence that the corresponding MHF has maximum cumulative memory cost. 
Nevertheless, proving (or disproving) that the CMC of a dMHF is captured by the cumulative cost of the optimal dynamic pebbling strategy for the underlying dynamic graph is still an open question that is outside the scope of the current work.

\paragraph{ciMHF Implementation Through Shuffling.} 
The only problem is that the above $k$-restricted dynamic graph is actually a data-dependent construction; once the input $x$ is fixed, the memory access patterns of the above construction is completely deterministic! 
Thus a side-channel attacker that obtains a memory access pattern will possibly be able to distinguish between future inputs. 
Our solution is to have a hidden random key $K$ for each separate evaluation of the password hash. 
The hidden random key $K$ does not alter the hash value of $x$ in any manner, so we emphasize that there is no need to know the value of the hidden key $K$ to perform computation. 
However, each computation using a separate value of $K$ induces a different memory access pattern, so that no information is revealed to side-channel attackers looking at locations of read/write instructions. 

Let $L$ be a set of the last $N$ consecutive nodes from our previous graph construction, which we suppose is called $G_0$. 
We form $G$ by appending a path of length $N$ to the end of $G_0$. 
We introduce a gadget that partitions the nodes in $L$ into blocks $B_1,B_2,\ldots,B_{N/k}$ of size $k$ each. 
We then enforce that for $i\in[N]$ and $j=i\mod{k}$, the $i\th$ node in the final $N$ nodes of $G$ has a parent selected uniformly at random from $B_{j+1}$, depending on the input $x$. 
Thus to compute the label of $i$, the evaluation algorithm should know the labels of all nodes in $B_{j+1}$. 

We allow the evaluation algorithm to manipulate the locations of these labels so that the output of the algorithm remains the same, but each computation induces a different memory access pattern. 
Specifically, the random key $K$ induces a shuffling of the locations of the information within each block of the block partition gadget. 
Thus if the size of each block is sufficiently large, then with high probability, two separate computations of the hash for the same password will yield distinct memory access patterns, effectively computationally data-independent. 
Then informally, a side-channel attacker will not be able to use the memory access patterns to distinguish between future inputs. 

In fact, this approach works for a general class of graphs satisfying a property that we call ``amenable to shuffling''. 
We characterize the properties of the dynamic graphs that are amenable to shuffling in \secref{sec:shuffle} and show that $k$-restricted dynamic graphs that are ameanble to shuffling can be used in the design of MHFs to yield computationally data-independent sequential evaluation algorithms. 
\begin{theorem}
\thmlab{thm:informal:shuffle}
For each DAG $G$ that is amenable to shuffling, there exists a computationally data-independent sequential evaluation algorithm computing a MHF based on the graph $G$ that runs in time $\O{N}$. 
(Informal, see \thmref{thm:shuffling:cimhf}.)
\end{theorem}
We believe that our techniques for converting graphs that are amenable to shuffling to ciMHFs may be of independent interest. 

Finally, we provide a version of our dMHF with $\Omega(N^2)$ cumulative memory complexity that is amenable to shuffling. 
Combining this maximally hard $k$-restricted dMHF using a DAG that is amenable to shuffling with \thmref{thm:informal:shuffle}, we obtain a maximally hard ciMHF.
\begin{theorem}
\thmlab{thm:informal:upper}
Let $0<\eps<1$ be a constant and $k=\Omega(N^\eps)$. 
Then there exists a family $\G$ of $k$-restricted graphs with $\cc(\G) = \Omega(N^2)$ that is amenable to shuffling.
\end{theorem}
We prove \thmref{thm:informal:upper} in \secref{sec:implementation}, introducing the necessary formalities for computationally data-independent memory hard functions and the underlying systems model. 
Our results in \thmref{thm:informal:lower} and \thmref{thm:informal:upper} almost completely characterize the spectrum of $k$-restricted graphs. 
In fact, for a graph $G$ drawn uniformly at random from our distribution $\G$ in \thmref{thm:informal:upper} and any pebbling strategy $S$, not only do we have $\EEx{G\sim\G}{\cc(S,G)}=\Omega(N^2)$, but we also have $\cc(S,G)=\Omega(N^2)$ with high probability. 
\section{Preliminaries}
We use the notation $[N]$ to denote the set $\{0,1,\ldots,N-1\}$. 
For two numbers $x$ and $y$, we use $x\circ y$ to denote their concatenation. 

Given a directed acyclic graph (DAG) $G=(V,E)$ and a node $v \in V$, we use $\parents_G(v)= \{u~:~(u,v) \in E\}$ to denote the parents of node $v$. 
We use $\ancestors_G(v) = \bigcup_{i \geq 1} \parents_G^i(v)$ to denote the set of all ancestors of $v$ --- here, $\parents_G^2(v) = \parents_G\left(\parents_G(v) \right)$ and $\parents^{i+1}_G(v) = \parents_G\left( \parents^i_G(v)\right)$. 
We use $\indeg(v) = \left| \parents(v)\right|$ to denote the number of incoming edges into $v$ and define $\indeg(G)=\underset{v\in V}{\max}\,\indeg(v)$. 
Given a set $S \subseteq V$, we use $G-S$ to refer to the graph obtained by deleting all nodes in $S$ and all edges incident to $S$. 
We use $\depth(G)$ to denote the number of nodes in the longest directed path in $G$. 

\begin{definition}
A DAG $G=(V,E)$ is \emph{$(e,d)$-reducible} if there exists a subset $S \subseteq V$ of size $|S|\le e$ such that any directed path $P$ in $G$ of length $d$ contains at least one node in $S$. 
We call such a set $S$ a \emph{depth-reducing set}. 
If $G$ is not $(e,d)$-reducible, then we say that $G$ is \emph{$(e,d)$-depth robust}. 
\end{definition}

For a DAG $G=(V=[N],E)$, we use $G_{\le i}$ to denote the subgraph of $G$ induced by $[i]$. 
In other words, $G_{\le i}=(V',E')$ for $V'=[i]$ and $E'=\{(a,b)\in E\,|\,a,b\le i\}$. 

\paragraph{The Parallel Random Oracle Model.}
We review the parallel random oracle model (pROM), as introduced by Alwen and Serbinenko~\cite{AlwenS15}. 
There exists a probabilistic algorithm $\A^{\H}$ that serves as the main computational unit, where $\A^{\H}$ has access to an arbitrary number of parallel copies of an oracle $\H$ sampled uniformly at random from an oracle set $\mathbb{H}$ and proceeds to do computation in a number of rounds. 
In each round $i$, $\A^{\H}$ maintains a state $\sigma_i$ along with initial input $x$. 
$\A^{\H}$ determines a batch of queries $\q_i$ to send to $\H$, receives and processes the responses to determine an updated state $\sigma_{i+1}$. 
At some point, $\A^{\H}$ completes its computation and outputs the value $\A^{\H}(x)$. 

We say that $\A^{\H}$ computes a function $f_{\H}$ on input $x$ with probability $\eps$ if $\PPr{\A^{\H}(x)=f_{\H}}\ge\eps$, where the probability is taken over the internal randomness of $\A$. 
We say that $\A^{\H}$ uses $t$ running time if it outputs $\A^{\H}(x)$ after round $t$. 
In that case, we also say $\A^{\H}$ uses space $\sum_{i=1}^t|\sigma_i|$ and that $\A^{\H}$ makes $q$ queries if $\sum_{i=1}^t\le q$. 

\paragraph{The Ideal Cipher Model}
In the ideal cipher model (ICM), there is a publicly available random block cipher, which has a $\kappa$-bit key $K$ and an $N$ bit input and output. 
Equivalently, all parties, including any honest parties and adversaries, have access to a family of $2^\kappa$ independent random permutations of $[N]$. 
Moreover for any given key $K$ and $x\in[N]$, both encryption $\Enc(K,x)$ and decryption $\Dec(K,x)$ queries can be made to the random block cipher.

\paragraph{Graph Pebbling.} 
The goal of the (black) pebbling game is to place pebbles on all sink nodes of some input directed acyclic graph (DAG) $G=(V,E)$. 
The game proceeds in rounds, and each round $i$ consists of a number of pebbles $P_i\subseteq V$ placed on a subset of the vertices. 
Initially, the graph is unpebbled, $P_0=\emptyset$, and in each round $i\ge 1$, we may place a pebble on $v\in P_i$ if either all parents of $v$ contained pebbles in the previous round ($\parents(v)\subseteq P_{i-1}$) or if $v$ already contained a pebble in the previous round ($v\in P_{i-1}$). 
In the sequential pebbling game, at most one new pebble can be placed on the graph in any round (i.e., $\left|P_i \backslash P_{i-1} \right| \leq 1)$, but this restriction does not apply in the parallel pebbling game. 

We use $\pPeb_G$ to denote the set of all valid parallel pebblings of a fixed graph $G$. 
The \emph{cumulative cost} of a pebbling $P=(P_1,\ldots,P_t) \in \pPeb_G$ is the quantity $\cc(P):=|P_1|+\ldots+|P_t|$ that represents the sum of the number of pebbles on the graph during every round. 
The (parallel) \emph{cumulative pebbling cost} of the fixed graph $G$, denoted $\pcc(G) := \min_{P \in \pPeb_G} \cc(P)$, is the cumulative cost of the best legal pebbling of $G$. 

\begin{definition}[Dynamic/Static Pebbling Graph]
We define a \emph{dynamic pebbling graph} as a distribution $\G$ over directed acyclic graphs $G = (V=[N],E)$ with edges $E=\{(i-1,i): i \leq N\} \cup  \{(r(i),i): i \leq N\}$, where $r(i) < i-1$ is a randomly chosen directed edge. 
We say that an edge $(r(i),i)$ is {\em dynamic} if $r(i)$ is not chosen until a black pebbled is place on node $i-1$. 
We say that the graph is \emph{static} if none of the edges are dynamic. 
\end{definition}
We now define a labeling of a graph $G$. 
\begin{definition}
\deflab{def:labeling}
Given a DAG $G=(V=[N],E)$ and a random oracle function $H:\Sigma^*\rightarrow \Sigma^w$ over an alphabet $\Sigma$, we define the labeling of graph $G$ as $L_{G,H}:\Sigma^*\rightarrow\Sigma^*$.
In particular, given an input $x$ the $(H,x)$ labeling of $G$ is defined recursively by
\[L_{G,H,x}(v)=
\begin{cases}
H(v\circ x),&\indeg(v)=0\\
H\left(v\circ L_{G,H,x}(v_1)\circ\cdots\circ L_{G,H,x}(v_d)\right),&\indeg(v)>0,
\end{cases}\]
where $v_1,\ldots,v_d$ are the parents of $v$ in $G$, according to some predetermined lexicographical order. 
We define $f_{G,H}(x)= L_{G,H,x}(s_1)\circ\ldots\circ L_{G,H,x}(s_k)$, where $s_1,\ldots,s_k$ are the sinks of $G$ sorted lexicographically by node index. 
If there is a single sink node $s_G$ then $f_{G,H}(x)=L_{G,H,x}(s_G)$. 
We omit the subscripts $G,H,x$ when the dependency on the graph $G$ and hash function $H$ is clear.
For a distribution of dynamic graphs $\G$, we say $f_{\G,H}(x)=f_{G,H}(x)$ once a dynamic graph $G$ has been determined from the choice of $H$ and $x$. 
\end{definition}

\noindent
For a node $i$, we define $\potential(i)$ to be set $Y_i$ of minimal size such that $\PPr{r(i)\in Y_i}=1$.
We now define $k$-restricted dynamic graphs, which can characterize both dMHFs and iMHFs.

\begin{definition}[$k$-Restricted Dynamic Graph]
We say that a dynamic pebbling graph $\G$ is $k$-restricted if for all $i$, $\potential(i)\le k$.
\end{definition}

Observe that $k=1$ corresponds to an iMHF while $k=N$ corresponds to a dMHF. 
Hence, $k$-restricted dynamic graphs can be viewed as spectrum between dMHFs and iMHFs. 

We define the cumulative cost of pebbling a dynamic graph similar to the definition of cumulative cost of pebblings on static graphs.
We first require the following definition of a dynamic pebbling strategy:

\begin{definition}
[Dynamic Pebbling Strategy]
A \emph{dynamic pebbling strategy} $S$ is a function that takes as input
\begin{enumerate}
\item an integer $i\leq N$
\item an initial pebbling configuration $P_0^i\subseteq[i]$ with $i\in P_0^i$
\item a partial graph $G_{\leq i+1}$
\end{enumerate}
The output of $S(i, P_0^i, G_{\leq i+1})$ is a legal sequence of pebbling moves $P_1^i,\ldots,P_{r_i}^i$ that will be used in the next phase, to place a pebble on node $i+1$, so that $i+1 \in P_{r_i}^i\subseteq [i+1]$. 
Given $G \sim \G$ we can abuse notation and write $S(G)$ for the valid pebbling produced by $S$ on the graph $G$ i.e., $P_1^0,\ldots,P_{r_0}^0,P_{1}^1,\ldots,P_{r_1}^1,\ldots,P_{r_{1}}^{N-1},\ldots,P_{r_{N-1}}^{N-1}$. 
Here, $P_1^i,\ldots,P_{r_i}^i = S(i,P_0^i, G_{\leq i+1})$ where $P_0^i = P_{r_{i-1}}^{i-1}$ and for $i=1$ we set $P_0^0 = \emptyset$. 
\end{definition}

We thus define $\cc(S,G)$ to the pebbling cost of strategy $S$ when we sample a dynamic graph $G$ and $\cc(S,\G)=\EEx{G\sim\G}{\cc(S,G)}$. 
Finally, we define $\cc(\G)=\min_S\cc(S,\G)$, where the minimum is taken over all dynamic pebbling strategies $S$. 
More generally, we define $\cc(S,\G,\delta) = \max\{k\,:\,\PPPr{G\in\G}{\cc(S,G)\ge k}\ge 1-\delta\}$. 
Fixing $\delta$ to be some negligible function of $N$, we can define $\cc_{\delta}(\G)=\min_S\cc(S,\G,\delta)$.
\section{General Attack Against $k$-Restricted Graphs}
\seclab{sec:generic}
In this section, we describe a general attack against $k$-restricted graphs. 
We show that the attack incurs cost $o(N^2)$ for $k=o\left(N^{1/\log\log N}\right)$, proving that there is no maximally memory hard $k$-restricted graph for small $k$.

We first require the following formulation of Valiant's Lemma, which shows the existence of a subroutine $\valiant(G,e,d)$ to find a depth-reducing set $S$ of size at most $e$ within a graph $G$, for $e=\frac{\eta\d N}{\log(N)-\eta}$ and $d=\frac{N}{2^\eta}$, where $\eta>0$. 

\begin{lemma}[Valiant's Lemma]
\cite{Valiant77}
\lemlab{lem:valiant}
For any DAG $G=(V,E)$ with $N$ nodes, indegree $\d$, and $\eta>0$, there exists an efficient algorithm $\valiant(G,e,d)$ to compute a set $S$ of size $|S|\le e:=\frac{\eta\d N}{\log(N)-\eta}$ such that $\depth(G-S)\le d:=\frac{N}{2^\eta}$.
\end{lemma}

The high level intuition of the generic attack is as follows. 
By Valiant's Lemma (\lemref{lem:valiant}), $G\sim\G$ is $(e,d)$-reducible for $e=\frac{\eta\d N}{\log(N)-\eta}$ and $d=\frac{N}{2^\eta}$. 
We will construct a dynamic pebbling strategy $\A$ that for all times $t$, maintains a depth-reducing set $S_t$ such that $\depth(G_t-S_t)\le d$, where $G_t$ is the portion of $G$ revealed after running $\A$ for time $t$, $G_t = G_{\leq i}$ for $i = 1+ \max \bigcup_{j=1}^t P_j$, where each $P_j\subseteq [N]$ represents the set of pebbled nodes during round $j$. 
Observe that for any $i$, $G_{\le i}$ is $(e,d)$-reducible and hence $G_t$ is also $(e,d)$-reducible for all times $t$. 
Thus, the depth reducing set $S_t$ has size at most $e$ for all times $t$ and can be computed by a subroutine $\valiant$, by \lemref{lem:valiant}.  
We now describe how $\A$ maintains this depth-reducing set through a series of \emph{light phases} and \emph{balloon phases}. 

We first set a parameter $g$ that we will eventually optimize. 
The goal of each light phase $i$ is to pebble the next $g$ nodes that have yet to be revealed. 
That is, if $x_i$ is the largest node for which $\A$ has placed a pebble at some point prior to light phase $i$, then the goal of light phase $i$ is to pebble the interval $[x_i+1,x_i+g]$. 
To begin light phase $i$ at some time $t_i$, we require that $(\potential([x_i+1,x_i+g])\cup S_{t_i})\subseteq P_{t_i}$ for some depth-reducing set $S_{t_i}$ of size at most $e$, such that $\depth(G_{t_i}-S_{t_i})\le d$. 
Once this pre-condition is met, then light phase $i$ simply takes $g$ steps to pebble $[x_i+1,x_i+g]$, since pebbles are already placed on $\potential([x_i+1,x_i+g])$. 
Hence, the post-condition of light phase $i$ at some time $u_i$ is pebbles on the node $x_i+g$ and some depth-reducing set $S_{u_i}$ of size at most $e$, such that $\depth(G_{u_i}-S_{u_i})\le d$. 

The goal of each balloon phase $i$ is to place pebbles on all revealed nodes of the graph, to meet the pre-condition of light phase $i+1$. 
To begin balloon phase $i$ at some time $r_i$, we first have a necessary pre-condition that pebbles are placed on some depth-reducing set $S_{r_i}$ of size at most $e$ such that $\depth(G_{r_i}-S_{r_i})\le d$. 
Once this pre-condition is met, then balloon phase $i$ simply takes $d$ steps to pebble the entire graph $G_{r_i}$, meeting the post-condition of balloon phase $i$. 

We now formally prove the cumulative memory complexity of the attack in \algref{alg:generic}.
\begin{theorem}
\thmlab{thm:attack}
Let $\G$ be any family of $k$-restricted dynamic graphs with constant $\indeg(\G)$. 
Then 
\[\cc(\G)=\O{\frac{N^2}{\log\log N}+N^{2-1/2\log\log N}\sqrt{k^{1-1/\log\log N}}}.\]
\end{theorem}
\begin{proof}
We analyze the cost of the pebbling strategy of \algref{alg:generic}. 
Since $G$ is drawn from a distribution of $k$-restricted dynamic graphs, then for any node $x_i$, $r(x_i)$ must be one of at most $k$ labels. 
Thus for any consecutive $g$ nodes, $|\potential([x_i+1,x_i+g])|\le gk$. 
Hence it suffices to keep $gk$ pebbles on the set of potential parents $\potential([x_i+1,x_i+g])$ to pebble the interval $[x_i+1,x_i+g]$, as well as a depth-reducing set of size at most $e$, for each of the $g$ steps during light phase $i$. 
On the other hand, balloon phase $i$ takes $d$ steps, each of which trivially contains at most $N$ pebbles. 

$\A$ proceeds using $\frac{N}{g}$ total rounds of light and balloon phases, by pebbling $g$ consecutive nodes at a time. 
Therefore, the total cost of the attack is $\O{Ngk+Ne+\frac{N}{g}\cdot dN}$, where the first and second terms originate from the cost of the light phases and the third term results from the cost of the balloon phases. 
Since we set $e=\frac{\eta\d N}{\log(N)-\eta}$ and $d=\frac{N}{2^\eta}$ from Valiant's Lemma (\lemref{lem:valiant}) so that the total cost is $\O{\frac{\eta\d N^2}{\log(N)-\eta} + Ngk + \frac{N^3}{2^\eta g}}$. 
By setting $g=\frac{N}{\sqrt{k2^\eta}}$, the total cost is $\O{\frac{N^2\eta\d}{\log N-\eta}+N^2\sqrt{\frac{k}{2^\eta}}}$. 
Finally, by setting $\eta=\frac{\log k + \log N - \log\d}{\log\log N}$, the total cost is $\O{\frac{N^2}{\log\log N}+N^{2-1/2\log\log N}\sqrt{k^{1-1/\log\log N}}}$. 
\end{proof}
Note that if $k=o(N^{1/\log\log N})$, then $\cc(\G)=o(N^2)$.
\begin{corollary}
\corlab{cor:attack}
Let $\G$ be any $k$-restricted dynamic graph with $k=o(N^{1/\log\log N})$ and constant indegree. 
Then $\cc(\G)=o(N^2)$.
\end{corollary}

\begin{algorithm}[!htb]
\caption{Generic pebbling strategy against dynamic DAG $G$.}
\alglab{alg:generic}
\begin{algorithmic}[1]
\Require{An integer $i$, an initial pebbling configuration $P_0^i\subseteq[i]$ with $i\in P_0^i$, a partial graph $G_{\le i+1}$, and parameters $d,e,g$.}
\Ensure{A legal pebbling of $G_{\le i+1}$.}
\State{$\invariant\gets\True$}
\If{$i\pmod{g}\equiv 0$ and $\depth(G_{\le i+1}-P_0^i)>d$}
\State{$\invariant\gets\False$}
\ElsIf{$\depth(G_{\le i+1}-P_0^i)>d$ or $\{i\}\cup\potential([i+1,i+g])\not\subseteq P_0^i$}
\State{$\invariant\gets\False$}
\EndIf
\If{$\invariant$}
\Comment{If pre-conditions met.}
\If{$i\pmod{g}\equiv 0$}
\Comment{Balloon phase}
\For{$j=1$ to $j=d$}
\State{$P_j^i=P_{j-1}^i\cup D_j$, where $D_j$ are the nodes at depth $d$ from $P_0^i$.}
\EndFor
\State{$P_{d+1}^i=\valiant(G_{\le i},e,d)\cup\potential([i+1,i+g])$.}
\Comment{See \lemref{lem:valiant}}
\Else
\Comment{Light phase}
\State{$P_1^i=P_0^i\cup\{i+1\}$.}
\EndIf
\Else
\Comment{If pre-conditions not met.}
\For{$j=1$ to $j=i+1$}
\State{$P_j^i=P_{j-1}^i\cup\{j\}$.}
\EndFor
\EndIf
\end{algorithmic}
\end{algorithm}
\section{$k$-Restricted Graphs with high CMC}
\seclab{sec:construct}
In this section, we describe a construction of $k$-restricted graphs with high cumulatively memory complexity that builds into our ultimate ciMHF implementation. 
We first describe the block partition extension gadget, which requires an input graph $G$ and outputs a family of $k$-restricted dynamic graphs. 
However, na\"{i}vely choosing the input graph $G$ does not yield a construction with high CMC. 

Intuitively, the block partition extension gadget takes the last $N$ nodes of $G$ and partitions them into $\frac{N}{k}$ blocks of $k$ nodes each. 
The gadget then creates $N$ more nodes in a path, such that a parent $r(j)$ of node $j$ in this path is drawn uniformly at random from block $i$, where $i=j\pmod{\frac{N}{k}}$. 
The intuition is that by drawing parents uniformly at random from each block in round robin fashion, we encourage an algorithm to keep $\Omega(N)$ nodes on the graph for $\Omega(N)$ steps. 
Of course, the graph could always maintain $o(n)$ pebbles on the graph and repebble when necessary, but we can discourage this strategy by making the repebbling procedure as expensive as possible. 

A first attempt would be to choose a highly depth-robust graph $G$, such as a grates graph, which informally has long paths of length $\Omega(N^{1-\eps})$ for any constant $0<\eps<1$, even when $\Omega(N)$ nodes are removed from $G$. 
Thus if an algorithm does not maintain $\Omega(N)$ pebbles on the graph, the repebbling strategy costs at least $\Omega(N^{2-\eps})$. 
Although this is a good start, this does not quite match the $\Omega(N^2)$ CMC of various dMHFs. 
We defer full discussion of how to increase the CMC to $\Omega(N^2)$ to later in this section.  

Instead, we first define a specific way to obtain a $k$-restricted dynamic graph given a graph $G$ with $N$ nodes and a parameter $k$.
\begin{definition}
[Block Partition Extension]
Given a DAG $G=(V=[\alpha N],E)$ with $\alpha N$ nodes containing a set of $O=[(\alpha-1)N+1,\alpha N]$ output nodes of size $N$ and a parameter $k$, let $O_i=[(\alpha-1)N+1+ik, (\alpha-1)N+(i+1)k]$ for $i\in\left[\frac{N}{k}\right]$ so that $\{O_i\}$ forms a partition of $O$. 
We define the \emph{block partition extension} of $G$, denoted $\partition_k(G)$, as a distribution of graphs $\G_{G,k}$. 
Each graph $G'$ sampled from $\G$ has vertices $V'=[(\alpha+1)N]$ and edges $E'=E\cup F$, where $F$ is defined as the edges $(i-1,i)$ and $(r(i),i)$ for each $i\in[\alpha N+1,(\alpha+1)N]$, where $r(i)$ is drawn uniformly at random from $O_{i\mod{\frac{N}{k}}}$. 
\end{definition}
An example of the block partition extension is given in \figref{fig:partition}.
\begin{figure}[!htb]
\centering
\begin{tikzpicture}
\filldraw[thick, top color=white,bottom color=green!50!] (0,0) rectangle+(0.5,1);
\filldraw[thick, top color=white,bottom color=green!50!] (0.5,0) rectangle+(0.5,1);
\draw (1,0) rectangle +(1,1);
\node at (0.25,0.5){$1$};
\node at (0.75,0.5){$2$};
\node at (1.5,0.5){$\ldots$};
\filldraw[thick, top color=white,bottom color=green!50!] (2,0) rectangle+(0.5,1);
\filldraw[thick, top color=white,bottom color=red!50!] (2.5,0) rectangle+(0.5,1);
\filldraw[thick, top color=white,bottom color=red!50!] (3,0) rectangle+(0.5,1);
\node at (2.25,0.5){$k$};
\draw (3.5,0) rectangle +(1,1);
\node at (4,0.5){$\ldots$};
\filldraw[thick, top color=white,bottom color=red!50!] (4.5,0) rectangle+(0.5,1);
\node at (5.5,0.5){$\ldots$};
\filldraw[thick, top color=white,bottom color=blue!50!] (6,0) rectangle+(0.5,1);
\filldraw[thick, top color=white,bottom color=blue!50!] (6.5,0) rectangle+(0.5,1);
\draw (7,0) rectangle +(1,1);
\node at (7.5,0.5){$\ldots$};
\filldraw[thick, top color=white,bottom color=blue!50!] (8,0) rectangle+(0.5,1);

\draw [decorate,decoration={brace}](2.45,-0.2) -- (0.05,-0.2);
\node at (1.25,-0.6){Block $1$}; 
\draw [decorate,decoration={brace}](4.95,-0.2) -- (2.55,-0.2);
\node at (3.75,-0.6){Block $2$};
\node at (5.5,-0.6){$\ldots$};
\draw [decorate,decoration={brace}](8.45,-0.2) -- (5.95,-0.2);
\node at (7.25,-0.6){Block $\frac{N}{k}$};

\draw (1,-2.5) circle [radius=0.3];
\draw (1.4,-2.5)[->] -- +(0.4,0);
\draw (2.2,-2.5) circle [radius=0.3];
\draw (2.6,-2.5)[->] -- +(0.4,0);
\draw (3.4,-2.5) circle [radius=0.3];
\draw (3.8,-2.5)[->] -- +(0.4,0);
\node at (4.6,-2.5){$\ldots$};
\draw (5,-2.5)[->] -- +(0.4,0);
\draw (5.8,-2.5) circle [radius=0.3];
\draw (6.2,-2.5)[->] -- +(0.4,0);
\draw (7,-2.5) circle [radius=0.3];

\draw (0.75,-1)[->] -- (2.2,-2.1);
\draw (3.75,-1)[->] -- (3.4,-2.1);
\draw (7.25,-1)[->] -- (5.8,-2.1);
\draw (1.75,-1)[->] -- (1.75,-1.5) -- (7,-1.5) -- (7, -2.1);
\end{tikzpicture}
\caption{Parent $r(i)$ is drawn uniformly at random from the nodes partitioned to each block.}
\figlab{fig:partition}
\end{figure}
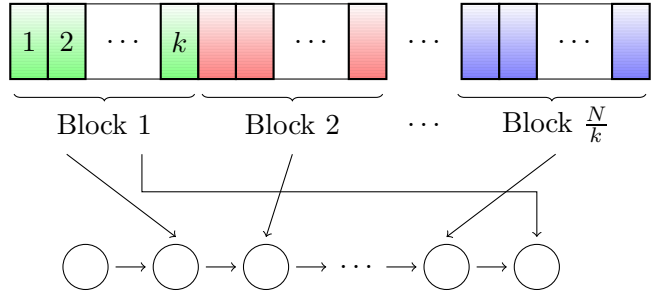

Our ultimate construction also requires the use of superconcentrator graphs, defined as follows:
\begin{definition}
A graph $G$ with $\O{N}$ vertices is a \emph{superconcentrator} if there exists an input set $I$ and an output set $O$ with $|I|=|O|=N$ such that for all $S_1\subseteq I,S_2\subseteq O$ with $|S_1|=|S_2|=k$, there are $k$ node disjoint paths from $S_1$ to $S_2$.
\end{definition}
It is known that there exists superconcentrators with $|I|=|O|=N$, constant indegree and $\O{N}$ total nodes~\cite{Pippenger77,LengauerT82}. 
We now show that a set $Y$, which contains more nodes than a set $S$ of removed nodes, has at least $N-|S|$ ancestors in $G-S$. 
\begin{lemma}
\lemlab{lem:sc:repebble}
Given a superconcentrator $G$ with $N$ input nodes $I$ and $N$ output nodes $O$, let $S$ and $Y\subseteq O$ be sets of nodes with $|S|<|Y|$. 
Then $|I\cap\ancestors_{G-S}(Y)|\ge N-|S|$. 
\end{lemma}
\begin{proof}
Let $X\subseteq I$ be the last $|Y|$ nodes of $I$. 
Since $G$ is a superconcentrator, then $G$ contains at least $|Y|$ node disjoint paths between $X$ and $Y$. 
Since $|S|<|Y|$, then one of these paths from $X$ to $Y$ that does not intersect $S$. 
Thus, $X$ contains some ancestor of $Y$ in $G-S$ and in fact by considering the paths associated with decreasing order of nodes in $X$, it follows that $|I\cap\ancestors_{G-S}(Y)|\ge N-|S|$. 
\end{proof}
We require the use of grates graphs $\{\grates_{N,\eps}\}_{N=1}^\infty$ \cite{Schnitger83}. 
For each constant $\eps > 0$ and each $N \geq 1$ the graph $\grates_{N,\eps} = (V_N,E_{N,\eps})$ has $\O{N}$ nodes and constant indegree $\indeg(\grates_{N,\eps}) = \O{1}$. 
Moreover, the graph $\grates_{N,\eps}$ contains source nodes $I_N \subset V_N$ and $N$ sinks $O_N \subset V_N$. 
Given a set $S \subseteq V_N$ of deleted nodes we say that an output node $y \in O_N$ is $c$-good with respect to $S$ if $\left|I_N \cap \ancestors_{\grates_{N,\eps}-S}(y) \right| \geq cN$ i.e., for at least $cN$ input nodes $x \in I_N$ the graph $\grates_{N,\eps}-S$ still contains a path from $x$ to $y$. 
The grates graph contains several properties summarized below.

\begin{theorem}
\cite{Schnitger83}
\thmlab{thm:grates}
For each $\eps >0$ there exist constants $\gamma,c > 0$ such that for all $N \geq 1$ the graph $\grates_{N,\eps}$ is $(\gamma N,c N^{1-\eps})$-depth robust. 
Furthermore, for each set $S \subseteq V_N$ of size $|S| \leq \gamma N$ at least $cN$ output nodes are still $c$-good with respect to $S$. Formally,
\[ \left| \left\{x \in O~:~ \left|I_N \cap \ancestors_{G-S}(x) \right| \geq cN \right\} \right| \geq  cN \ . \]
\end{theorem}

\noindent
We require the use of graph overlays, defined as follows:
\begin{definition}[Graph overlays]
Given a DAG $H=(V=[N],E)$ with sources $I=\{1,\ldots,n_1\}$ and sinks $O=\{N-n_2+1,\ldots,N\}$, a DAG $G_1=(V_1=[n_1],E_1)$, and a DAG $G_2=(V_2=[n_2],E_2)$, we define:
\begin{enumerate}
\item
the \emph{graph overlay} $G'=\overlay(G_1,H,G_2)$ by $G'=([N],E')$, where $(i,j)\in E'$ if and only if $(i,j)\in E$ or $(i,j)\in E_1$ or $(i+N-n_2,j+N-n_2)\in E_2$ 
\item
the \emph{superconcentrator overlay} of an $N$ node DAG $G$ by $\superconc(G)=\overlay(G,\SC_N,L_N)$, where $\SC_N$ is a superconcentrator with $N$ input (sources) and output (sinks) nodes and $L_N$ is the line graph of $N$ nodes 
\item
the \emph{grates overlay} of an $N$ node DAG $G$ by $\grates_{\eps}(G)=\overlay(G,\grates_{N,\eps},L_N)$.
\end{enumerate}
\end{definition}
An example of a graph overlay is displayed in \figref{fig:overlay}.

\begin{figure}
\centering
\begin{tikzpicture}
\node at (-1,0){$G_2$};
\node at (-1,1.75){$G$};
\node at (-1,3.5){$G_1$};

\draw (0,0) circle [radius=0.3];
\draw (0.4,0)[->] -- +(0.7,0);
\draw (1.5,0) circle [radius=0.3];
\draw (1.9,0)[->] -- +(0.7,0);
\draw (3,0) circle [radius=0.3];

\draw (0,1) circle [radius=0.3];
\draw (1.5,1) circle [radius=0.3];
\draw (3,1) circle [radius=0.3];
\draw (0,2.5) circle [radius=0.3];
\draw (1.5,2.5) circle [radius=0.3];
\draw (3,2.5) circle [radius=0.3];
\draw (0,2.1)[->] -- (0,1.4);
\draw (0.3,2.1)[->] -- (1.2,1.4);
\draw (1.5,2.1)[->] -- (1.5,1.4);
\draw (2.7,2.1)[->] -- (1.8,1.4);
\draw (3,2.1)[->] -- (3,1.4);

\draw (0,3.5) circle [radius=0.3];
\draw (0.4,3.5)[->] -- +(0.7,0);
\draw (1.5,3.5) circle [radius=0.3];
\draw (1.9,3.5)[->] -- +(0.7,0);
\draw (3,3.5) circle [radius=0.3];
\draw [->] (0,3.9) to [out=20,in=160] (3,3.9);

\draw (4,1.75)[->] -- +(1,0);

\draw (6,1) circle [radius=0.3];
\draw (7.5,1) circle [radius=0.3];
\draw (9,1) circle [radius=0.3];
\draw (6,2.5) circle [radius=0.3];
\draw (7.5,2.5) circle [radius=0.3];
\draw (9,2.5) circle [radius=0.3];
\draw (6,2.1)[->] -- (6,1.4);
\draw (6.3,2.1)[->] -- (7.2,1.4);
\draw (7.5,2.1)[->] -- (7.5,1.4);
\draw (8.7,2.1)[->] -- (7.8,1.4);
\draw (9,2.1)[->] -- (9,1.4);

\draw (6.4,1)[->] -- +(0.7,0);
\draw (7.9,1)[->] -- +(0.7,0);
\draw (6.4,2.5)[->] -- +(0.7,0);
\draw (7.9,2.5)[->] -- +(0.7,0);
\draw [->] (6,2.9) to [out=20,in=160] (9,2.9);

\end{tikzpicture}
\caption{An example of a graph overlay $\overlay(G_1,G,G_2)$.}
\figlab{fig:overlay}
\end{figure}
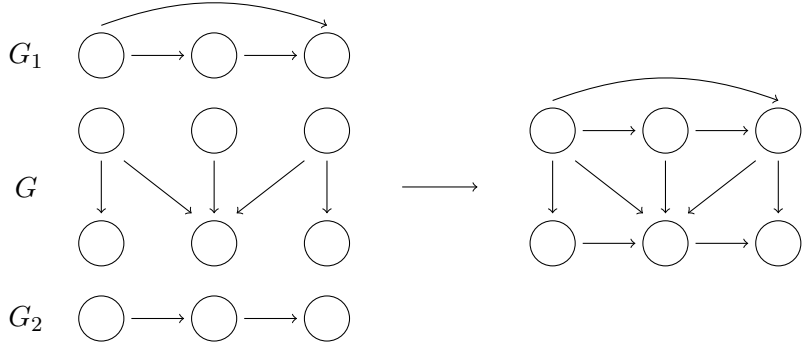

We describe a preliminary attempt at a ciMHF construction in \figref{fig:graph:random}. 
At a high level, the construction consists of four components. 
The first component is a grates graph $G_1$ with $N$ nodes. 
The second component is a superconcentrator overlay with $\O{N}$ nodes, including $N$ input nodes and $N$ output nodes, so that $G_2=\superconc(G_1)$. 
The third component consists of a grates overlay with $\O{N}$ nodes including $N$ output nodes, so that $G_3=\grates_{\eps}(G_2)$. 
The $N$ output nodes of $G_3$ are partitioned into $\frac{N}{k}$ blocks, each with $k$ nodes, in preparation for a block partition extension in the final component. 
Namely, the fourth component consists of a $k$-restricted graph with $N$ nodes, so that $G_4=\partition_k(G_3)$.

\begin{figure}[!htb]
\begin{mdframed}
Sampling Algorithm, for $k=\Omega(N^\eps)$:
\begin{enumerate}
\item
$G_1=\grates_{N,\eps}$ 
\item
$G_2=\superconc(G_1)$
\item
$G_3=\grates_{\eps}(G_2)$
\item
$G_4\sim\partition_k(G_3)$
\end{enumerate}
\end{mdframed}
\vspace{-0.5cm}
\caption{First attempt at ciMHF. Each parent $r(i)$ is randomly chosen from the labels in specific block corresponding to $i$.}
\figlab{fig:graph:random}
\end{figure}

The intuition for the $\Omega(N^2)$ cumulative pebbling complexity is as follows. 
Suppose there exists a time $t_{\bad}$ with a ``small'' number of pebbles on the graph. 
Then with high probability, walking a pebble $s=\frac{N}{4k}$ steps on the final layer of the graph will require some number of output nodes of the grates graph to be repebbled. 
Again with high probability, repebbling one of these output nodes requires a large number of input nodes of the grates graph to be repebbled. 
These input nodes are the output nodes of the superconcentrator at the second layer. 
The superconcentrator property then implies that $\Omega(N)$ nodes of the grates graph on the first layer will need to be repebbled. 
For a grates graph that is $(\Omega(N),\Omega(N^{1-\eps}))$-depth robust, this cost is at least $\Omega(N^{2-\eps})$ every $s$ steps. 
Thus, the total cost is at least $\min\left(\Omega(N^2), k\Omega(N^{2-\eps})\right)$, which is just $\Omega(N^2)$ for $k=\Omega(N^\eps)$.

We now show that our construction in \figref{fig:graph:random} has cumulative memory complexity $\Omega(N^2)$. 
\begin{theorem}
\thmlab{thm:cc:graph:random}
Let $G$ be drawn from the distribution of $k$-restricted graphs in \figref{fig:graph:random}, for $k=\Omega(N^\eps)$. 
There exist constants $c_1>0$ and $c_2\in(0,1)$ such that for any dynamic pebbling strategy $S$,
\[\PPPr{\G}{\cc(S,G)>c_1N^2}\ge 1-c_2^{N/k}.\]
\end{theorem}
\begin{proof}
Let $\alpha N$ be the total number of nodes in $G_3$ so that the total number of nodes in $G$ is $(\alpha+1)N$. 
Let $x,y\in(0,1)$ be constants such that the grates graph $G_1$ is $(xN,yN^{1-\eps})$-depth robust. 
By \thmref{thm:grates}, there exist constants $0<c<\frac{x}{2}$ and $0<\gamma<c$ such that for any set $S$ with $|S|\le\gamma N$, at least $cN$ nodes in the output nodes of $G_3$ are $c$-good with respect to $S$. 
For each node $i$, let $t_i$ be the first time that node $i$ is pebbled. 
Suppose there exists a time $t_{\bad}$ with $t_i\le t_{\bad}<t_{i+1}$ such that there are $|P_{t_\bad}|<\frac{\gamma N}{4}$ pebbles on the graph. 

For a node $j$ in the output set $[(\alpha-1)N,\alpha N]$ of $G_3$, we call an index $j$ a \emph{costly index} if $j$ is $c$-good with respect to $P_{t_\bad}$ and let $\costly$ be the set of costly indices. 
Note that if a node $i\in\costly$, then by definition $i\notin P_{t_\bad}$. 
By \thmref{thm:grates} and the observation that $|P_{t_\bad}|<\frac{\gamma N}{4}$, there are at least $cN$ nodes in the output set of $G_3$ are $c$-good with respect to $P_{t_\bad}$, i.e., $|\costly|>cN$. 
Then for $s:=\frac{N}{k}$, we call $j\in[i, i+s]$ a \emph{missed costly index} if $r(j)\notin P_{t_\bad}$ and let $r(j)\in\costly$. 

For each $j\in\left[\frac{N}{k}\right]$, let $c_j$ be the number of costly indices in block $j$ of the output set of $G_3$, i.e., $c_j:=\left|\costly\cap[(\alpha-1)N+(j-1)k+1,(\alpha-1)N+jk]\right|$. 
Since the parents of $[i,i+s]$ are exactly one random node from each of the $\frac{N}{k}$ blocks, then the probability $p$ that no parent of $[i,i+s]$ is a missed costly index is 
\[p:=\prod_{j=1}^{N/k}\left(1-\frac{c_j}{k}\right)\le\left(\frac{k}{N}\sum_{j=1}^{N/k}\left(1-\frac{c_j}{k}\right)\right)^{N/k},\]
where the inequality holds by the Arithmetic Mean-Geometric Mean Inequality. 
Since $\sum c_j=cN$, then 
\begin{align*}
p&\le\left(1-\frac{k}{N}\sum_{j=1}^{N/k}\frac{c_j}{k}\right)^{N/k}\le(1-c)^{N/k}.
\end{align*}
Thus with high probability, there will be some missed costly index. 
By \lemref{lem:sc:repebble} and the definition of $c$-good, any missed costly index requires $cN$ nodes in the input set of $G_3$ to be repebbled.
Since the input set of $G_3$ is connected by a superconcentrator to the output set of $G_1$, the $cN$ nodes in the input set of $G_3$ that need to be repebbled have at least $N-cN$ ancestors in the output set of $G_1$. 
Thus, at least $N-cN-|P_{t_\bad}|$ nodes in $G_1$ must be repebbled. 

Because $G_1$ is $(xN,yN^{1-\eps})$-depth robust, then $G_1-S$ is $(xN-|S|,yN^{1-\eps})$-depth robust for any set $S$. 
Moreover, the cost to pebble $G_1-S$ is at least $(xN-|S|)(yN^{1-\eps})$. 
In particular, if $G_1-S$ is the set of nodes in $G_1$ must be repebbled, then it costs at least $(xN-cN-|P_{t_\bad}|)(yN^{1-\eps})$ to repebble $G_1-S$. 
Since $c<\frac{x}{2}$ and $|P_{t_\bad}|<\frac{\gamma N}{4}<\frac{cN}{4}<\frac{xN}{8}$, then the cost is at least $\frac{3xy}{8} N^{2-\eps}$. 

Hence to pebble an interval $[i,i+s]$ with $s=\frac{N}{k}$, either $\frac{\gamma N}{4}$ pebbles are kept on the graph for all $s=\frac{N}{k}$ steps or if we at any point in time $j \in [i,i+s]$  we have $|P_j| \leq \gamma N/4$ then (whp) an the pebbling algorithm incurs cost $\frac{3xy}{8} N^{2-\eps}$ to repebble the graph during the next $s$ steps $[j,j+s]$. By partitioning the last $N$ nodes of the graph $G$ into $k$ disjoint intervals of length $\frac{N}{k}$, it follows that the total cost is at least $\min\left(\frac{\gamma N^2}{4}, \frac{3}{16}xykN^{2-\eps}\right)$. 
Thus for $k=\Omega(N^\eps)$, the total cost is $\Omega(N^2)$ with high probability. 
\end{proof}
\begin{corollary}
Let $G$ be drawn from the distribution of $k$-restricted graphs in \figref{fig:graph:random}, for $k=\Omega(N^\eps)$. 
Then $\cc(G)=\Omega(N^2)$.
\end{corollary}
\section{$k$-Restricted Graphs: Amenable to Shuffling}
\seclab{sec:shuffle}
In this section, we introduce a useful property for certain dynamic graphs: amenable to shuffling.  
In \secref{sec:implementation}, we will describe computationally data-independent evaluation algorithms for evaluating memory hard function based on dynamic graphs that are amenable to shuffling. 

\subsection{Characterization of Dynamic Graphs Amenable to Shuffling}
We first describe the properties of dynamic graphs that are amenable to shuffling. 
Recall that for a node $i$, we define $\potential(i)$ to be set $Y_i$ of minimal size such that $\PPr{r(i)\in Y_i}=1$, where $r(i) < i-1$ is randomly chosen so that the directed edge $(r(i),i)$ is in the dynamic graph.  
\begin{definition}[Amenable to Shuffling]
\deflab{def:shuffling}
Let $G$ be a DAG with $\alpha N$ nodes for some constant $\alpha>1$ and let $L$ be the last $N$ nodes of $G$. 
Suppose that $L$ can be partitioned into $\frac{N}{k}$ groups $G_1,\ldots,G_{\frac{N}{k}}$ such that
\begin{enumerate}
\item
Uniform Size of Groups: $|G_i|=k$ for all $i\in\left[\frac{N}{k}\right]$.
\item
Large Number of Potential Parents: For each $v\in L$, $|\potential(u)|=k$. 
\item
Potential Parents not in $L$: For each $v\in L$, $\potential(u)\subseteq [(\alpha-1)N]$
\item
Same Potential Parents for Each Group: For all $i\in\left[\frac{N}{k}\right]$ and $u,v\in G_i$, $\potential(u)=\potential(v)$. 
\item
Different Potential Parents for Different Groups: For all $i,j\in\left[\frac{N}{k}\right]$ with $i\neq j$, let $u\in G_i$ and $v\in G_j$. 
Then $\potential(u)\cap\potential(v)=\emptyset$. 
\item
No Collision for Parents: For each $i\in\left[\frac{N}{k}\right]$, define the event $\unique_i$ to be the event that $r(u)\neq r(v)$ for all $u,v \in G_i$. 
Then $\PPr{\unique_i}=1$ for all $i\in\left[\frac{N}{k}\right]$.
\item
Data-Independency: The subgraph induced by the first $(\alpha-1)N$ nodes is a static graph.
\end{enumerate}
Then we call $G$ \emph{amenable to shuffling}.
\end{definition}
We shall show in \thmref{thm:shuffling:cimhf} in \secref{sec:implementation} that dynamic graphs that are amenable to shuffling can be used for memory hard functions with computationally data-independent evaluation algorithms. 
We now describe a version of \figref{fig:graph:random} that is amenable to shuffling. 

\subsection{Version of Construction Amenable to Shuffling}
\seclab{sec:cr}
To ensure that there does not exist $i\neq j$ such that $r(i)=r(j)$, we slightly modify the construction of block partition extensions to the concept of a collision-resistant block partition extension. 
For the sake of presentation, note that we use $2N$ output nodes in $G$ in the following definition. 
\begin{definition}
[Collision-Resistant Block Partition Extension]
Given a DAG $G=(V=[\alpha N],E)$ with $\alpha N$ nodes containing a set of $O=[(\alpha-2)N+1,\alpha N]$ output nodes of size $2N$ and a parameter $k$, let $O_i=[(\alpha-2)N+1+2ik, (\alpha-2)N+2(i+1)k]$ for $i\in\left[\frac{N}{k}\right]$ so that $\{O_i\}$ forms a partition of $O$. 
We define the \emph{collision-resistant block partition extension} of $G$, denoted $\crpartition_k(G)$, as a distribution of graphs $\G_{G,k}$. 
Each graph $G'$ sampled from $\G$ has vertices $V'=[(\alpha+1)N]$ and edges $E'=E\cup F$, where $F$ is defined as the edges $(i-1,i)$ and $(r(i),i)$ for each $i\in[\alpha N+1,(\alpha+1)N]$, where $r(i)$ is defined as follows:
\begin{enumerate}
\item
Let $\Enc$ be the family of all permutations of $\left[2k\right]$, so that for each fixed $j$, 
\[\{\Enc(j,\ell)\}_{\ell\in\left[2k\right]}=\left[2k\right].\]
\item
For each $i\in[\alpha N+1,(\alpha+1)N]$, let $j=i\mod{\frac{N}{k}}$ and define $1\le p\le k$ to be the unique integer such that $i=\frac{N}{k}(p-1) +\alpha N i+j$. 
Then we define $r(i)=(\alpha-2)N+1+2jk+\Enc(x\circ j,p)$, so that $r(i)\in O_j$. 
\end{enumerate}
\end{definition}
Observe that the collision-resistant block partition extension is a gadget that yields a dMHF, since the parent function $r(i)$ has the key $x\circ j$ to its permutation function $\Enc$. 
Hence, the underlying dynamic graph differs across different input values $x$. 

Then our construction of the ciMHF appears in \figref{fig:graph:permutation} and \figref{fig:final:graph}. 
As before, the construction consists of four layers. 
The first layer consists of a grates graph with $2N$ nodes, $G_1=\grates_{2N,\eps}$. 
The second layer consists of a superconcentrator overlay with $\O{N}$ nodes with $2N$ input nodes and $2N$ output nodes, so that $G_2=\superconc(G_1)$. 
The third layer consists of a grates overlay with $\O{N}$ nodes including $2N$ output nodes, so that $G_3=\grates_{\eps}(G_2)$. 
The $2N$ output nodes of $G_3$ are partitioned into $\frac{N}{k}$ blocks, each with $2k$ nodes, which allows the final layer to be a $2k$-restricted grpah. 
In particular, the fourth layer uses a collision-free block partition extension rather than the block partition extension of \figref{fig:graph:random}. 
\begin{figure}[!htb]
\begin{mdframed}
Sampling algorithm, for $k=\Omega(N^\eps)$:
\begin{enumerate}
\item
$G_1=\grates_{2N,\eps}$
\item
$G_2=\superconc(G_1)$ 
\item
$G_3=\grates_{\eps}(G_2)$
\item
$G_4\sim\crpartition_k(G_3)$
\end{enumerate}
\end{mdframed}
\vspace{-0.5cm}
\caption{Second attempt at ciMHF. Each parent $r(i)$ is chosen by a permtuation of the labels in specific block corresponding to $i$. The underlying graph is visualized in \figref{fig:final:graph}.}
\figlab{fig:graph:permutation}
\end{figure}

\begin{figure}
\centering
\begin{tikzpicture}
\draw (-0.5,-0.5) rectangle +(9,2.5);
\node at (6.5,0.7){$\crpartition$};

\draw (0,0) circle [radius=0.3];
\draw (0.4,0)[->] -- +(0.7,0);
\draw (1.5,0) circle [radius=0.3];
\draw (1.9,0)[->] -- +(0.7,0);
\node at (3,0){$\ldots$};
\draw (3.4,0)[->] -- +(0.7,0);
\draw (4.5,0) circle [radius=0.3];

\draw (0,1.5) circle [radius=0.3];
\draw (0.4,1.5)[->] -- +(0.7,0);
\draw (1.5,1.5) circle [radius=0.3];
\draw (1.9,1.5)[->] -- +(0.7,0);
\node at (3,1.5){$\ldots$};
\draw (3.4,1.5)[->] -- +(0.7,0);
\draw (4.5,1.5) circle [radius=0.3];

\draw (3,1.3)[->] -- (3,0.2);
\draw (2.5,1.3) to [out=270,in=180] (3,0.75);
\draw (3,0.75)[->] to [out=0,in=90] (3.5,0.2);
\draw (3.5,1.3) to [out=270,in=0] (3,0.75);
\draw (3,0.75)[->] to [out=180,in=90] (2.5,0.2);

\draw (-3,1) rectangle +(9,2.5);
\node at (-2,2.2){$\grates_{\eps}$};

\draw (0,3) circle [radius=0.3];
\draw (0.4,3)[->] -- +(0.7,0);
\draw (1.5,3) circle [radius=0.3];
\draw (1.9,3)[->] -- +(0.7,0);
\node at (3,3){$\ldots$};
\draw (3.4,3)[->] -- +(0.7,0);
\draw (4.5,3) circle [radius=0.3];

\draw (3,2.8)[->] -- (3,1.7);
\draw (2.5,2.8) to [out=270,in=180] (3,2.25);
\draw (3,2.25)[->] to [out=0,in=90] (3.5,1.7);
\draw (3.5,2.8) to [out=270,in=0] (3,2.25);
\draw (3,2.25)[->] to [out=180,in=90] (2.5,1.7);

\draw (-0.5,2.5) rectangle +(9,2.5);
\node at (7,3.7){$\superconc$};

\draw (0,4.5) circle [radius=0.3];
\draw (0.4,4.5)[->] -- +(0.7,0);
\draw (1.5,4.5) circle [radius=0.3];
\draw (1.9,4.5)[->] -- +(0.7,0);
\node at (3,4.5){$\ldots$};
\draw (3.4,4.5)[->] -- +(0.7,0);
\draw (4.5,4.5) circle [radius=0.3];

\draw (3,4.3)[->] -- (3,3.2);
\draw (2.5,4.3) to [out=270,in=180] (3,3.75);
\draw (3,3.75)[->] to [out=0,in=90] (3.5,3.2);
\draw (3.5,4.3) to [out=270,in=0] (3,3.75);
\draw (3,3.75)[->] to [out=180,in=90] (2.5,3.2);

\draw (-3,4) rectangle +(9,2.5);
\node at (-2,5.2){$\grates$};

\draw (0,6) circle [radius=0.3];
\draw (0.4,6)[->] -- +(0.7,0);
\draw (1.5,6) circle [radius=0.3];
\draw (1.9,6)[->] -- +(0.7,0);
\node at (3,6){$\ldots$};
\draw (3.4,6)[->] -- +(0.7,0);
\draw (4.5,6) circle [radius=0.3];

\draw (3,5.8)[->] -- (3,4.7);
\draw (2.5,5.8) to [out=270,in=180] (3,5.25);
\draw (3,5.25)[->] to [out=0,in=90] (3.5,4.7);
\draw (3.5,5.8) to [out=270,in=0] (3,5.25);
\draw (3,5.25)[->] to [out=180,in=90] (2.5,4.7);
\end{tikzpicture}
\caption{Final construction of \figref{fig:graph:permutation}.}
\figlab{fig:final:graph}
\end{figure}

We now show that the construction of \figref{fig:graph:permutation} has cumulative cost $\Omega(N^2)$ with high probability. 
The proof is almost verbatim to \thmref{thm:cc:graph:random} except that the graph overlays now have $2N$ input and output nodes. 
\begin{theorem}
\thmlab{thm:cc:permutation}
Let $0<\eps<1$ be a constant and $k=\Omega(N^\eps)$. 
Let $\G$ be drawn from the distribution of $2k$-restricted graphs in \figref{fig:graph:permutation}. 
There exist constants $c_1>0$ and $c_2\in(0,1)$ such that for any dynamic pebbling strategy $S$,
\[\PPPr{G\in\G}{\cc(S,G)>c_1N^2}\ge 1-c_2^{N/k}.\]
\end{theorem}
\begin{corollary}
Let $\G$ be drawn from the distribution of $2k$-restricted graphs in \figref{fig:graph:permutation}, for $k=\Omega(N^\eps)$. 
Then $\cc(\G)=\Omega(N^2)$.
\end{corollary}
Finally, we observe that the construction of \figref{fig:graph:permutation} is amenable to shuffling since it satisfies the properties of \defref{def:shuffling}.
\section{Implementation of ciMHF}
\seclab{sec:implementation}
In this section, we describe how to implement our construction in a way that is computationally data-independent. 
We first formalize the notion of computationally data-independent and then describe the system model we utilize. 

\subsection{Computationally Data-Independent MHF (ciMHF) and Systems Model}
We define the security of a computationally data-independent memory hard function in terms of the following game: a side-channel attacker $\A$ selects two inputs $x_0,x_1$ and sends these inputs to an honest party $\H$. 
We first require the following definition of leakage patterns.

\paragraph{Leakage Pattern.}
We define the leakage pattern of an evaluation algorithm $\mhfeval$ by the sequence of request and store instructions made in each round. 
Specifically, in each round $r$, an attacker can observe from the leakage pattern the blocks of memory to be loaded into cache, as requested by $\mhfeval$. 
Let $i=(i_1,\ldots,i_m)$ be the sequence of locations of all blocks requested by $\mhfeval$ in a particular round $r$ through some command $\load(i)$.  
If $i$ is completely contained in cache, then no events will be observed by the attacker. 
Otherwise, if $i$ is not completely contained in cache, we use $\request_r$ to denote the locations of the blocks in memory, as well as their sizes, requested by $\mhfeval$ in round $r$. 
Similarly, we use $\store_r$ to denote the locations of the blocks, as well as their sizes, stored into memory by $\mhfeval$ in round $r$.  
We do not allow the attacker to observe the contents of the requested or stored blocks. 
Formally, the leakage pattern $\LP$ is the information $\{(\request_r,\store_r)\}_{r=1}^t$ and is dependent on the algorithm $\mhfeval$, random oracle $H$, internal randomness $R$, and input value $x$. 

\paragraph{Computationally Data-Independency Game.}
$\H$ runs a (randomized) evaluation algorithm $\mhfeval$ on both inputs $x_0$ and $x_1$, yielding two leakage patterns $\LP_0$ and $\LP_1$, where $\LP_i$ for $i\in\{0,1\}$ depends on both the input $x_i$ and the random coins selected during the execution of $\mhfeval$. 
$\H$ then picks a random challenge bit $b \in \{0,1\}$ and sends $\LP_b,\LP_{1-b}$ to $\A$ to simulate a side-channel. 
The goal of $\A$ is to predict $b$ i.e., match each input with the corresponding leakage pattern. 
For a secure ciMHF we guarantee that {\em any} PPT side-channel attacker $\A$ wins the game with only negligible advantage over random guessing. 

Formally, the game consists of three phases \setup, \challenge, and \guess, which are described as follows.  
\begin{mdframed}
\textbf{Data independency game for ciMHF:}
\begin{description}
\item[\setup]
In this phase, $\A$ selects the security parameter $ \lambda $ and two challenge messages $x_0$ and $x_1$ and sends them to $\H$. 
Here we assume without loss of generality that the runtime of the evaluation algorithm $\mhfeval$ on $x_0$ and $x_1$ are the same. 

\item[\challenge]
In this phase, $\H$ selects a random bit $b\in\{0,1\}$ and random coins $R_0, R_1\in\{0,1\}^{\lambda}$ uniformly at random and then samples $\lp_0 \gets \LP(\mhfeval(x_0; R_0)) $ and $ \lp_{1} \gets \LP(\mhfeval(x_1; R_1)) $.
$\H$ sends the ordered pair $(\lp_b,\lp_{1-b})$ to $\A$. 

\item[\guess]
After receiving $(\lp_b, \lp_{1-b})$, the adversary $\A$ outputs $b'$ as a guess for $b$. 
The adversary wins the game if $b=b'$.
\end{description}
\end{mdframed}

The advantage of the adversary to win the game of computationally data independency of the given MHF is defined as 
\begin{align*}
Adv_{\A, \MHF}^{\mathsf{ind-lp-iMHF}} = \left|\frac{1}{2}-\Pr[\A(x_0, x_1,\lp_b,\lp_{1-b})=b': b=b')]\right|,
\end{align*}
where $\lp_i=\LP(x_i; \mhfeval(x_i; R_i))$.

\begin{definition}[Computational data independency] \deflab{singleround}
An evaluation algorithm $\mhfeval$ is computationally data independent if for all non-uniform circuits $\A= \{\A_{\lambda} \}_{\lambda \in\mathbb{N}}$, there is a negligible function $\negl(\cdot)$ such that $Adv_{\A,\MHF}^{\mathsf{ind-lp-iMHF}}<\negl(\lambda)$.
\end{definition} 
With the proper random coins, a memory-hard function with an evaluation algorithm that satisfies the above definition reveals only a negligible amount of information through its leakage patterns, and we thus call such a function a computationally data-independent memory hard function. 

\paragraph{On the Definition of Computational Data Independency.}  In section \secref{multiround} we show that the definition is equivalent to a multi round version of the game in which the attacker can adaptively select the challenge $x_{i,0},x_{i,1}$ in each round $i \leq r$ after observing $\lp_{i-1,b},\lp_{i-1,1-b}$ --- the memory access patterns from the last round. We also prove that the two security notions are asymptotically equivalent when the attacker runs in polynomial time --- in terms of concrete security parameters we lose a factor of $r$ (number of challenge rounds) in the reduction. 

We are primarily motivated by the password hashing application where the inputs $x_0$ and $x_1$ come from a small domain, as user selected passwords tend to have low entropy~\cite{SP:Bonneau12}. In practice it is reasonable to assume that $r$ is polynomial i.e., if the user only authenticates $\poly(\lambda)$ times then there are at most $r=\poly(\lambda)$ memory access patterns for the attacker to observe. Assuming that the input domain has size $\poly(\lambda)$ a brute-force attacker cannot use the leaked memory access pattern on input $x$ to eliminate any candidate password $x'$ with high probability, otherwise the attacker could have used the pair $x$ and $x'$ to win the data independency game. 

However, in settings where the input domain is very large and $r$ is super-polynomial it will be better to adopt a concrete security definition (see \secref{multiround}). The asymptotic definition in \defref{singleround} does not definitively rule out the possibility that an attacker can substantially narrow the search space after many (super-polynomial) side channel attacks. For example, suppose that the attacker gets to observe $\lp_i \gets \LP(\mhfeval(x; R_i))$ for $i = 1,\ldots, 2^\lambda$, i.e., $2^\lambda$ independent evaluations of $\MHF$ on secret input $x$. Supposing that $Adv_{\A,\MHF}^{\mathsf{ind-lp-iMHF}}=2^{-\lambda}$ and that the input domain for $\MHF$ has size $2^{2\lambda}$, it is possible that each $\lp_i$ allows the attacker to eliminate a random subset of $Adv_{\A,\MHF}^{\mathsf{ind-lp-iMHF}} \times 2^{2\lambda}=2^\lambda$ candidate inputs, allowing the attacker to find $x$ after just $\O{\lambda 2^{\lambda}}$ examples. However, in practice it will usually be reasonable to assume that the attacker gets to observe $\lp_i $ a polynomial number of times i.e., the honest party will execute $\LP(\mhfeval(x; R_i))$ at most $\poly(\lambda)$ times.

\paragraph{Memory Architecture Assumptions.}
We consider a tiered random access memory architecture with main memory (RAM) and working memory (cache). 
We assume that main memory (RAM) is a shared resource with other untrusted processes, each of which have their own cache.  
Although the operating system kernel will enforce memory separation, i.e., only our program has some region of memory and that other processes cannot read/write to this block, it is also possible that an untrusted process will be able to infer the memory address of read/write operations in RAM (due to side-channel effects). 

Formally, the system allows programs access to two operations $\Write(i,x)$, which takes an address $i$ within the memory allocated to the program and writes the value $x$ at address $i$, and $\Read(i)$ which loads the data at location $i$. 
When an operation requests memory at location $i$, there are two possible outcomes. 
Either the data item is already in cache or the data item is not in cache. 
In the second case, the location of the item in memory is revealed through the leakage pattern. 
Hence, the leakage pattern is either $\bot$, if the data item is already in cache, or $i$, if the data item is not in cache.

\paragraph{Cache Replacement Policy.}
We now show that our implementation of the dynamic pebbling construction with cumulative memory cost $\Omega(N^2)$ is computationally data-independent. 
In particular, we provide an evaluation algorithm whose leakage pattern is computationally indistinguishable under each of the following cache replacement policies:
\begin{description}
\item[Least recently used (LRU)]
This policy tracks the most recent time each item in cache was used and discards the least recently used items first when cache is full and items need to be replaced.
\item[First in first out (FIFO)]
This policy evicts the first item that was loaded into cache, ignoring how recent or how often it has been accessed.
\end{description}

\subsection{ciMHF Implementation}
Recall from the definition of a graph labeling in \defref{def:labeling} that given a function $H$ and a distribution of dynamic graphs $\G$, the goal is to compute $f_{\G,H}(x)$ for some input $x$, which is equivalent to $f_{G,H}(x)$ once the graph $G$ has been determined by the choice of $H$ and $x$.  
In this section, we describe a computationally data-independent implementation of the construction of \secref{sec:cr}. 

We implement the first three layers as data-independent components. 
Namely, the grates graph $G_1$, its superconcentrator overlay $G_2$ and the subsequent grates overlay $G_3$ can be implemented deterministically. 
Observe that $G_3$ has $\alpha N$ nodes, including $2N$ output nodes $O$ that are partitioned into $\frac{N}{k}$ blocks of size $2k$ each. 
Specifically $O=O_1\cup\ldots\cup O_{\frac{N}{k}}$, where $O_j=[(\alpha-2)N+1+2jk, (\alpha-2)N+2(j+1)k]$ for $j\in\left[\frac{N}{k}\right]$. 

As stated in \figref{fig:graph:permutation}, the collision-resistant block partition extension $G_4$ is actually data-independent, since for each $i\in [(\alpha-1)N,\alpha N]$, each parent $r(i)$ of $i$ is chosen uniformly at random from $k$ possible nodes, but the random procedure is independent of the input $x$. 
Hence, the challenge is to implement a computationally data-independent version of the collision-resistant block partition extension. 
We demand the input of a key $K$ for each computation of $f_{G,H}(x)$. 
The value of $f_{G,H}(x)$ remains the same across all keys $K$ but the leakage pattern is different for each $K$. 

\paragraph{Data-Dependent Dynamic Graph.} 
For each $i\in [\alpha N+1,(\alpha+1)N]$, let $j=i\mod{\frac{N}{k}}$. 
To implement a computationally data-independent version of $G_4$ from \figref{fig:graph:permutation}, we use the value of $L_{G,H,x}(i-1)$ to select a previously ``unused'' node of $O_j$ as the parent $r(i)$ for $i$ that is not $i-1$. 
Here, we say a node $v\in O_j$ is unused if it is not the parent of any node besides $v+1$ and $i$. 

Since $L_{G,H,x}(i-1)$ can be viewed as a random integer modulo $2k$, we can use $L_{G,H,x}(i-1)\pmod{2k}$ as an input to a permutation with key $K$ to randomly choose the parent of $i$ from $O_j$. 
Observe that $i=\alpha N+j$ is the first time a parent will be selected from $O_j$. 
Moreover, observe that $m=L_{G,H,x}(i-1)\pmod{2k}+1$ can be viewed as a random number from $[2k]$ so we set $r(i)=m+(\alpha-2)N+1+2jk$ as the $m\th$ entry of $O_j$. 

Now if $L_{G,H,x}(i-1)\pmod{2k}$ were all unique across the values of $\{i|i\mod{\frac{N}{k}}\equiv j\}$, then there would be no collisions among selections of parents in $O_j$ and we would be done. 
However, since these values are not unique, we must do a little more work to avoid collisions, which would reveal information through leakage patterns about the parents of two nodes being the same. 
To ensure there are no collisions in $O_j$ among parents, we store an array $U_j$ of size $2k$ for each block $O_j$. 
For each $1\le\ell\le 2k$, we initialize $U_j[\ell]=\ell$. 
The purpose of the $U_j$ array is to ensure that the nodes of $O_j$ that have already appeared as parents are at the end of $U_j$. 
In the above example when $r(i)=m+(\alpha-2)N+1+2jk$, we then set $U_j[m]=2k$ and $U_j[2k]=m$. 
Then in the next round of selecting a parent from $O_j$, we choose uniformly at random from the first $2k-1$ entries of $U_j$ and in general, for the $s\th$ round of selecting a parent from $O_j$, we choose uniformly at random from the first $2k-s+1$ entries of $U_j$. 

Specifically for some $2\le s\le k$, consider the $s\th$ iteration in which a parent is selected from $O_j$. 
That is, for $i=\alpha N + j + (s-1)\frac{N}{k}$, the parent $r(i)$ is the $s\th$ parent among the nodes of $O_j$. 
Observe that $m=L_{G,H,x}(i-1)\pmod{2k-s+1}$ can be viewed as a random number from $[2k-s+1]$ and so $U_j[m]$ is a random entry among the unselected $2k-s+1$ nodes of $O_j$. 
We then swap the values of $U_j[m]$ and $U_j[2k-s+1]$ so that if $U_j[m]=a$ and $U_j[2k-s+1]=b$ previously then we set $U_j[m]=b$ and $U_j[2k-s+1]=a$. 
Hence, the invariant remains that the first $2k-s$ locations of $U_j$ have been unused. 

\paragraph{Shuffling Leakage Patterns.} 
Finally, we point out that the leakage pattern across all computations of $f_{G,H}(x)$ is still the same, since we have not actually incorporated the key $K$ in any of the above details. 
In summary, the above the description ensures a collision-resistant block partition extension that is data-dependent, but is still vulnerable to side-channel attacks. 
Hence, we add a final element to our implementation that shuffles the locations of each node $p\in O_j$ inside $O_j$. 
That is, for each $1\le j\le\frac{N}{k}$, we use the keyed permutation $\Enc$ to store the label of $p\in O_j$ in the location that corresponds to $\Enc(j\circ K,p)$ instead. 
Thus if $r(i)=p\in O_j$ for some node $i$, the algorithm must look at the location associated with $\Enc(j\circ K,p)$ to learn the value of $L_{G,H,x}(p)$. 
Therefore, the underlying graph $G$ is a dynamic graph that is data-dependent but the leakage pattern across each computation of $f_{G,H}(x)$ is different due to the choice of $K$ that shuffles the locations of all labels in each block $O_j$. 
A high level example of this shuffling is shown in \figref{fig:partition:permutation}.

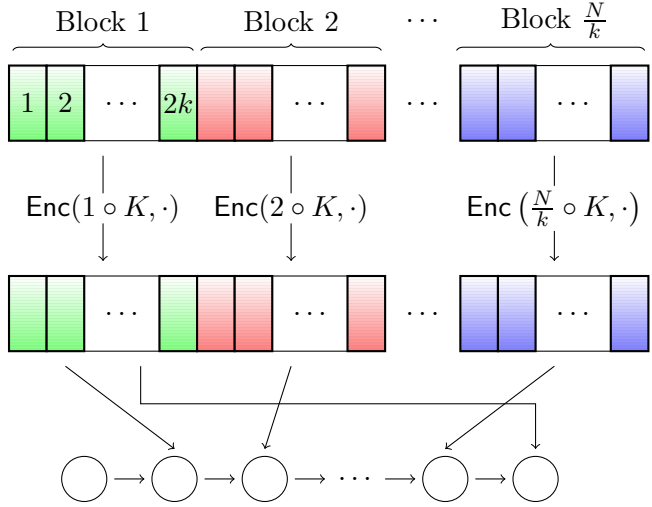
\begin{figure}
\centering
\begin{tikzpicture}
\filldraw[thick, top color=white,bottom color=green!50!] (0,2) rectangle+(0.5,1);
\filldraw[thick, top color=white,bottom color=green!50!] (0.5,2) rectangle+(0.5,1);
\draw (1,2) rectangle +(1,1);
\node at (0.25,2.5){$1$};
\node at (0.75,2.5){$2$};
\node at (1.5,2.5){$\ldots$};
\filldraw[thick, top color=white,bottom color=green!50!] (2,2) rectangle+(0.5,1);
\filldraw[thick, top color=white,bottom color=red!50!] (2.5,2) rectangle+(0.5,1);
\filldraw[thick, top color=white,bottom color=red!50!] (3,2) rectangle+(0.5,1);
\node at (2.25,2.5){$2k$};
\draw (3.5,2) rectangle +(1,1);
\node at (4,2.5){$\ldots$};
\filldraw[thick, top color=white,bottom color=red!50!] (4.5,2) rectangle+(0.5,1);
\node at (5.5,2.5){$\ldots$};
\filldraw[thick, top color=white,bottom color=blue!50!] (6,2) rectangle+(0.5,1);
\filldraw[thick, top color=white,bottom color=blue!50!] (6.5,2) rectangle+(0.5,1);
\draw (7,2) rectangle +(1,1);
\node at (7.5,2.5){$\ldots$};
\filldraw[thick, top color=white,bottom color=blue!50!] (8,2) rectangle+(0.5,1);

\draw [decorate,decoration={brace}](0.05,3.2) -- (2.45,3.2);
\node at (1.25,3.6){Block $1$}; 
\draw [decorate,decoration={brace}](2.55,3.2) -- (4.95,3.2);
\node at (3.75,3.6){Block $2$};
\node at (5.5,3.6){$\ldots$};
\draw [decorate,decoration={brace}](5.95,3.2) -- (8.45,3.2);
\node at (7.25,3.6){Block $\frac{N}{k}$};

\draw (1.25,1.8) -- (1.25,1.3);
\node at (1.25, 1.1){$\Enc(1\circ K,\cdot)$};
\draw (1.25,0.9)[->]  -- (1.25,0.4);
\draw (3.75,1.8) -- (3.75,1.3);
\node at (3.75, 1.1){$\Enc(2\circ K,\cdot)$};
\draw (3.75,0.9)[->]  -- (3.75,0.4);
\draw (7.25,1.8) -- (7.25,1.4);
\node at (7.25, 1.1){$\Enc\left(\frac{N}{k}\circ K,\cdot\right)$};
\draw (7.25,0.8)[->]  -- (7.25,0.4);

\filldraw[thick, top color=white,bottom color=green!50!] (0,-0.8) rectangle+(0.5,1);
\filldraw[thick, top color=white,bottom color=green!50!] (0.5,-0.8) rectangle+(0.5,1);
\draw (1,-0.8) rectangle +(1,1);
\node at (1.5,-0.3){$\ldots$};
\filldraw[thick, top color=white,bottom color=green!50!] (2,-0.8) rectangle+(0.5,1);
\filldraw[thick, top color=white,bottom color=red!50!] (2.5,-0.8) rectangle+(0.5,1);
\filldraw[thick, top color=white,bottom color=red!50!] (3,-0.8) rectangle+(0.5,1);
\draw (3.5,-0.8) rectangle +(1,1);
\node at (4,-0.3){$\ldots$};
\filldraw[thick, top color=white,bottom color=red!50!] (4.5,-0.8) rectangle+(0.5,1);
\node at (5.5,-0.3){$\ldots$};
\filldraw[thick, top color=white,bottom color=blue!50!] (6,-0.8) rectangle+(0.5,1);
\filldraw[thick, top color=white,bottom color=blue!50!] (6.5,-0.8) rectangle+(0.5,1);
\draw (7,-0.8) rectangle +(1,1);
\node at (7.5,-0.3){$\ldots$};
\filldraw[thick, top color=white,bottom color=blue!50!] (8,-0.8) rectangle+(0.5,1);

\draw (1,-2.5) circle [radius=0.3];
\draw (1.4,-2.5)[->] -- +(0.4,0);
\draw (2.2,-2.5) circle [radius=0.3];
\draw (2.6,-2.5)[->] -- +(0.4,0);
\draw (3.4,-2.5) circle [radius=0.3];
\draw (3.8,-2.5)[->] -- +(0.4,0);
\node at (4.6,-2.5){$\ldots$};
\draw (5,-2.5)[->] -- +(0.4,0);
\draw (5.8,-2.5) circle [radius=0.3];
\draw (6.2,-2.5)[->] -- +(0.4,0);
\draw (7,-2.5) circle [radius=0.3];

\draw (0.75,-1)[->] -- (2.2,-2.1);
\draw (3.75,-1)[->] -- (3.4,-2.1);
\draw (7.25,-1)[->] -- (5.8,-2.1);
\draw (1.75,-1)[->] -- (1.75,-1.5) -- (7,-1.5) -- (7, -2.1);
\end{tikzpicture}
\caption{Parent $r(i)$ is drawn uniformly at random from the nodes partitioned to each block.}
\figlab{fig:partition:permutation}
\end{figure}

Observe that this shuffling must be done completely in the cache to avoid leaking locations of labels during the shuffling. 
Hence, we require cache eviction policies such as the least recently used (LRU) or first in first out (FIFO) cache eviction policies to ensure that the entire block $O_j$ will remain in cache as the shuffling is performed. 
We describe the implementation in full in \figref{fig:implementation}. 

\begin{figure}[!htb]
\begin{mdframed}
Computationally data-independent sequential evaluation algorithm $\mhfeval(x;R)$ to compute $f_{G,H}(x)$ for any $k$-restricted dynamic graph $G$ that is amenable to shuffling.
\begin{enumerate}
\item
Data-independent phase:
\begin{enumerate}
\item
Let $G$ be a $k$-restricted dynamic graph with $\alpha N$ nodes for some constant $\alpha>1$ that is amenable to shuffling, $L$ be the last $N$ nodes of $G$, and $H$ be an arbitrary hash function. 
\item
Let $K=\mathsf{Setup}(1^{\lambda};R)$ be a hidden random permutation key for each computation of $f_{G,H}(x)$, given the security parameter $\lambda$ and random bits $R$. 
\item
Recall that the subgraph induced by the first $(\alpha-1)N$ nodes is a static graph. 
Compute the label $L_{G,H,x}(v)$ for each node $v\in(G-L)$.  
\end{enumerate}
\item
Shuffling phase:
\begin{enumerate}
\item
Since $G$ is amenable to shuffling, $L$ can be partitioned into groups $G_1,\ldots,G_{\frac{N}{k}}$ that satisfy the definition of \defref{def:shuffling}. 
For each $j\in\left[\frac{N}{k}\right]$, let $O_j=\potential(G_j)$.  
\item
For each $j\in\left[\frac{N}{k}\right]$, shuffle the contents of $O_j$:
\begin{enumerate}
\item
Let $v_1,\ldots,v_{k}$ be the vertices in $O_j$. 
\item
Load the labels $L_{G,H,x}(v_1),\ldots,L_{G,H,x}(v_{k})$ into cache. 
\item
Shuffle the positions of $L_{G,H,x}(v_1),\ldots,L_{G,H,x}(v_{k})$ so that for each $p\in[k]$, $L_{G,H,x}(v_p)$ is in the location that previously corresponded to $L_{G,H,x}(v_q)$, where $q=\Enc(j\circ K,p)$, where $\Enc$ is a keyed pseudorandom permutation of $k$. 
\end{enumerate}
\end{enumerate}
\item
Data-dependent phase:
\begin{enumerate}
\item
For each $j\in\left[\frac{N}{k}\right]$, initialize an array $U_j$ such that for all $1\le\ell\le k$, $U_j[\ell]=\ell$.
\item
For each $i=\alpha N+1$ to $(\alpha+1)N$:
\begin{enumerate}
\item
Let $j$ and $s$ be defined so that $1\le s\le k$ and $1\le j\le\frac{N}{k}$ given $i=\alpha N + j + (s-1)\frac{N}{k}$ and let $m=L_{G,H,x}(i-1)\pmod{k-s+1}+1$.
\item
Set $r(i)=U_j[m]+(\alpha-1)N+1+jk$ so that $r(i)\in O_j$ and load $L_{G,H,x}(r(i))$. 
(Recall that the label of $r(i)$ is actually located at the position where the label of node $\Enc(j\circ K,U[m])$ was previously located prior to the shuffling.)
\item
Load $L_{G,H,x}(r(i))$ and $L_{G,H,x}(i-1)$ and compute $L_{G,H,x}(i)=H(i\circ L_{G,H,x}(r(i))\circ L_{G,H,x}(i-1))$. 
\item
Let $U_j[U_j[m]]=a$ and $U_j[k-s+1]=b$. 
Then swap the values of $U_j$ at $U_j[m]$ and $k-s+1$ so that $U_j[U_j[m]]=b$ and $U_j[k-s+1]=a$.
\end{enumerate} 
\end{enumerate} 
\end{enumerate}
\end{mdframed}
\vspace{-0.5cm}
\caption{Description of evaluation algorithm for $k$-restricted graphs that are amenable to shuffling. Note that each computation of $f_{G,H}(x)$ requires as input random bits $R$ to generate the leakage patterns.}
\figlab{fig:implementation}
\end{figure}

\paragraph{A Note on Oblivious RAM.}
The complications with the cache eviction policies and shuffling leakage patterns originate from the necessity of not divulging information in the data-independency game. 
One reasonable question is whether these complications can be avoided with other implementations that conceal the leakage patterns. 
For example, an algorithm using oblivious RAM (ORAM), introduced by Goldreich and Ostrovsky~\cite{GoldreichO96}, reveals no information through the memory access patterns about the underlying operations performed. 
Thus, an algorithm using an ORAM data structure to evaluate a memory hard function would induce a computationally independent memory hard function, regardless of whether the underlying function is data-dependent or data-independent. 
\cite{GoldreichO96} describe an oblivious RAM simulator that transforms any program in the standard RAM model into a program in the oblivious RAM model, where the
leakage pattern is information theoretically hidden, which is ideal for the data-independency game.

Existing constructions of ORAM protocols such as Path ORAM~\cite{StefanovDSCFRYD18} require amortized $\Omega(\log N)$ bandwidth overhead. 
Hence given any dMHF and evaluation algorithm running in sequential time $M$, we can use ORAM to develop a new evaluation algorithm with a concealed leakage pattern, running in sequential time $N=M\log M$.
However, this is not ideal because the cumulative memory complexity of the dMHF is $\O{M^2}=\O{\frac{N^2}{\log^2 N}}$. 
Viewed in this way, the ciMHF construction is worse than known iMHF constructions that achieve CMC $\Omega\left(\frac{N^2}{\log N}\right)$ such as DRSample~\cite{AlwenBH17, BlockiHKLXZ19}. 
In fact, even for $k$-restricted graphs, we still obtain a blow-up of $\Omega(\log^2 K)$, which is $\Omega\left(\frac{\log^2 N}{\log^2\log N}\right)$ when $k=\Omega(N^{1/\log\log N})$. 
Otherwise for $k=o(N^{1/\log\log N})$, our dynamic pebbling attack in \corref{cor:attack} shows that the CMC is at most $o(N^2)$. 

Although Boyle and Naor~\cite{BoyleN16} proposed the notion of \emph{online} ORAM, where the operations to be performed arrive in an online manner, and observe that the lower bounds of~\cite{GoldreichO96} do not hold for online ORAM, Larsen and Nielsen~\cite{LarsenN18} answer this open question by proving an amortized $\Omega(\log N)$ bandwidth overhead lower bound on the bandwidth of any online ORAM. 
Therefore, it does not seem obvious how to use ORAM in the implementations of maximally hard ciMHFs.

\subsubsection{Implementation and Analysis.}
\begin{figure}[!htb]
\begin{mdframed}
Hybrid:
\begin{enumerate}
\item
Data-independent phase:
\begin{enumerate}
\item
Let $G$ be a $k$-restricted dynamic graph with $\alpha N$ nodes for some constant $\alpha>1$ that is amenable to shuffling, $L$ be the last $N$ nodes of $G$, and $H$ be an arbitrary hash function. 
\item
Let $K=\mathsf{Setup}(1^{\lambda};R)$ be a hidden random permutation key for each computation of $f_{G,H}(x)$, given the security parameter $\lambda$ and random bits $R$. 
\item
Recall that the subgraph induced by the first $(\alpha-1)N$ nodes is a static graph. 
Compute the label $L_{G,H,x}(v)$ for each node $v\in(G-L)$.  
\end{enumerate}
\item
Shuffling phase:
\begin{enumerate}
\item
Since $G$ is amenable to shuffling, $L$ can be partitioned into groups $G_1,\ldots,G_{\frac{N}{k}}$ that satisfy the definition of \defref{def:shuffling}. 
For each $j\in\left[\frac{N}{k}\right]$, let $O_j=\potential(G_j)$.  
\item
For each $j\in\left[\frac{N}{k}\right]$, shuffle the contents of $O_j$:
\begin{enumerate}
\item
Let $v_1,\ldots,v_{k}$ be the vertices in $O_j$. 
\item
Load the labels $L_{G,H,x}(v_1),\ldots,L_{G,H,x}(v_{k})$ into cache. 
\item
Shuffle the positions of $L_{G,H,x}(v_1),\ldots,L_{G,H,x}(v_{k})$ so that for each $p\in[k]$, $L_{G,H,x}(v_p)$ is in the location that previously corresponded to $L_{G,H,x}(v_q)$, where $q=\Enc(j\circ K,p)$, where $\Enc$ is a keyed truly random permutation of $k$. 
\end{enumerate}
\end{enumerate}
\item
Data-dependent phase:
\begin{enumerate}
\item
For each $j\in\left[\frac{N}{k}\right]$, initialize an array $U_j$ such that for all $1\le\ell\le k$, $U_j[\ell]=\ell$.
\item
For each $i=\alpha N+1$ to $(\alpha+1)N$:
\begin{enumerate}
\item
Let $j$ and $s$ be defined so that $1\le s\le k$ and $1\le j\le\frac{N}{k}$ given $i=\alpha N + j + (s-1)\frac{N}{k}$ and let $m=L_{G,H,x}(i-1)\pmod{k-s+1}+1$.
\item
Set $r(i)=U_j[m]+(\alpha-1)N+1+jk$ so that $r(i)\in O_j$ and load $L_{G,H,x}(r(i))$. 
(Recall that the label of $r(i)$ is actually located at the position where the label of node $\Enc(j\circ K,U[m])$ was previously located prior to the shuffling.)
\item
Load $L_{G,H,x}(r(i))$ and $L_{G,H,x}(i-1)$ and compute $L_{G,H,x}(i)=H(i\circ L_{G,H,x}(r(i))\circ L_{G,H,x}(i-1))$. 
\item
Let $U_j[U_j[m]]=a$ and $U_j[k-s+1]=b$. 
Then swap the values of $U_j$ at $U_j[m]$ and $k-s+1$ so that $U_j[U_j[m]]=b$ and $U_j[k-s+1]=a$.
\end{enumerate} 
\end{enumerate} 
\end{enumerate}
\end{mdframed}
\vspace{-0.5cm}
\caption{Description of hybrid. Differs from \figref{fig:implementation} in that the hidden input key is used to index into the entire family of random permutations, rather than a pseudorandom permutation.}
\figlab{fig:hybrid}
\end{figure}

We require the hybrid in \figref{fig:hybrid} to argue that our implementation of \figref{fig:graph:permutation} is a ciMHF. 
The hybrid in \figref{fig:hybrid} differs from the implementation of \figref{fig:graph:permutation} in \figref{fig:implementation} in that the hidden input key is used to index from the entire family of random permutations, rather than a pseudorandom permutation. 
Thus the only way an adversary can distinguish between the hybrid and the real world sampler is by distinguishing between a random permutation and a pseudorandom permutation. 
On the other hand, if an adversary fails to distinguish between the hybrid and the real world sampler, then the cumulative memory complexity of the implementation requires $\Omega(N^2)$ since the leakage pattern of the hybrid is statistically equivalent to the dMHF construction in \figref{fig:graph:permutation}, where each parent is chosen a priori using a permutation drawn uniformly at random.

\begin{theorem}
\thmlab{thm:shuffling:cimhf}
For each DAG $G$ that is amenable to shuffling, there exists a computationally data-independent sequential evaluation algorithm $\mhfeval(x;R)$ computing the function $f_{\G, H}$ in time $\O{N}$. 
\end{theorem}
\begin{proof}
Consider the evaluation function in \figref{fig:implementation}. 
Observe that the hybrid in \figref{fig:hybrid} has the same distribution of leakage patterns as the dMHF of \figref{fig:graph:permutation}. 
Moreover, under the least recently used (LRU) or first in first out (FIFO) cache eviction policies, if $k$ is less than the size of the cache, then all the shuffling can be performed so an attacker observing the leakage patterns of the hybrid has no advantage in the data-independency game. 
Furthermore, the ciMHF implementation in \figref{fig:implementation} only differs from the hybrid in \figref{fig:hybrid} in the implementation of $\Enc$ as a pseudorandom permutation compared to a truly random permutation. 
Therefore, an attacker observing leakage patterns from the implementation in \figref{fig:implementation} only obtains a negligible advantage $\negl(\lambda)$ in the security parameter $\lambda$, in the data-independency game. 
Hence, the implementation of \figref{fig:implementation} is a ciMHF.
\end{proof}

\noindent
We now show that the evaluation function in \figref{fig:implementation} of the dMHF in \figref{fig:graph:permutation} is a maximally hard ciMHF. 
\begin{theorem}
\thmlab{thm:cc:cimhf}
Let $0<\eps<1$ be a constant and $k=\Omega(N^\eps)$. 
Then there exists a family $\G$ of $k$-restricted graphs with $\cc(\G) = \Omega(N^2)$ that is amenable to shuffling. 
Moreover, there exists a negligible value $\delta=\negl(N)$ such that $\cc_{\delta}(\G)=\Omega(N^2)$. 
\end{theorem}
\begin{proof}
Consider the evaluation function in \figref{fig:implementation} of the dMHF in \figref{fig:graph:permutation}. 
For the sake of completeness, the full implementation is also shown in \figref{fig:implementation:max}. 
Since the construction of \figref{fig:graph:permutation} is amenable to shuffling, then the evaluation algorithm is a ciMHF by \thmref{thm:shuffling:cimhf}. 
Finally by \thmref{thm:cc:permutation}, $\cc(\G) = \Omega(N^2)$. 

In fact, \thmref{thm:cc:permutation} implies that for $G\in\G$ drawn uniformly at random and any pebbling strategy $S$, not only is $\EEx{G\sim\G}{\cc(S,G)}=\Omega(N^2)$, but also $\cc(S,G)=\Omega(N^2)$ with probability at least $1-c^{N/k}$ for some constant $0<c<1$.
Thus for $\delta=1-c^{N/k}$, we have $\cc(S,\G,\delta)=\Omega(N^2)$ for any pebbling strategy $S$ and so $\cc_{\delta}(\G)=\Omega(N^2)$. 
\end{proof}

For the sake of completeness, we give the evaluation algorithm for the maximally hard ciMHF in \figref{fig:implementation:max}.
\begin{figure}[H]
\begin{mdframed}
Computationally data-independent sequential evaluation algorithm $\mhfeval(x;R)$
\begin{enumerate}
\item
Let $H$ be an arbitrary hash function and $K=\mathsf{Setup}(1^{\lambda};R)$ be a hidden random permutation key for each computation of $f_{\G,H}(x)$, given the security parameter $\lambda$ and random bits $R$. Let $k=\Omega(N^\eps)$.  
\item
Data-independent phase:
\begin{enumerate}
\item
$G_1=\grates_{2N,\eps}$
\item
$G_2=\superconc(G_1)$ 
\item
$G_3=\grates_{\eps}(G_2)$, which has $\alpha N$ total nodes, including $2N$ output nodes $O=O_1\cup\ldots\cup O_{\frac{N}{k}}$, where $O_j=[(\alpha-2)N+1+2ik, (\alpha-2)N+2(j+1)k]$ for $j\in\left[\frac{N}{k}\right]$. 
\item
Compute the label $L_{G_3,H,x}(v)$ for each node $v\in G_3$.
\end{enumerate}
\item
Shuffling phase:
\begin{enumerate}
\item
Let $G$ be the graph $G_3$ appended with $N$ additional nodes, so that $V=[(\alpha+1)N]$, and edges $(i-1,i)$ for each $\alpha N+1\le i\le(\alpha+1)N$.
\item
For each $j\in\left[\frac{N}{k}\right]$, shuffle the contents of $O_j$:
\begin{enumerate}
\item
Let $v_1,\ldots,v_{2k}$ be the vertices in $O_j$. 
\item
Load the labels $L_{G,H,x}(v_1),\ldots,L_{G,H,x}(v_{2k})$ into cache. 
\item
Shuffle the positions of $L_{G,H,x}(v_1),\ldots,L_{G,H,x}(v_{2k})$ so that for each $p\in[2k]$, $L_{G,H,x}(v_p)$ is in the location that previously corresponded to $L_{G,H,x}(v_q)$, where $q=\Enc(j\circ K,p)$, where $\Enc$ is a keyed pseudorandom permutation of $2k$. 
\end{enumerate}
\end{enumerate}
\item
Data-dependent phase:
\begin{enumerate}
\item
For each $j\in\left[\frac{N}{k}\right]$, initialize an array $U_j$ such that for all $1\le\ell\le 2k$, $U_j[\ell]=\ell$.
\item
For each $i=\alpha N+1$ to $(\alpha+1)N$:
\begin{enumerate}
\item
Let $j$ and $s$ be defined so that $1\le s\le k$ and $1\le j\le\frac{N}{k}$ given $i=\alpha N + j + (s-1)\frac{N}{k}$ and let $m=L_{G,H,x}(i-1)\pmod{2k-s+1}+1$.
\item
Set $r(i)=U_j[m]+(\alpha-2)N+1+2jk$ so that $r(i)\in O_j$ and load $L_{G,H,x}(r(i))$. 
(Recall that the label of $r(i)$ is actually located at the position where the label of node $\Enc(j\circ K,U[m])$ was previously located prior to the shuffling.)
\item
Load $L_{G,H,x}(r(i))$ and $L_{G,H,x}(i-1)$ and compute $L_{G,H,x}(i)=H(i\circ L_{G,H,x}(r(i))\circ L_{G,H,x}(i-1))$. 
\item
Let $U_j[U_j[m]]=a$ and $U_j[2k-s+1]=b$. 
Then swap the values of $U_j$ at $U_j[m]$ and $2k-s+1$ so that $U_j[U_j[m]]=b$ and $U_j[2k-s+1]=a$.
\end{enumerate} 
\end{enumerate} 
\end{enumerate}
\end{mdframed}
\vspace{-0.5cm}
\caption{Description of implementation of \emph{maximally hard} ciMHF. Again note that each computation of $f_{G,H}(x)$ requires as input random bits $R$ to generate the leakage pattern.}
\figlab{fig:implementation:max}
\end{figure}

\subsubsection{Extension to Multiple Rounds.} \seclab{multiround}
Finally, we show that our ciMHF implementation is robust to multiple rounds of leakage by considering a data independency game where an adversary is allowed to submit and observe multiple adaptive queries before outputting a guess for the hidden challenge bit $b$.
The game again consists of the phases \setup, \challenge, and \guess, which are described as follows.  
\begin{mdframed}
\textbf{Adaptive data independency game for ciMHF:}
\begin{description}
\item[\setup]
In this phase, $\A$ selects the security parameter $\lambda$ and sends it to $\H$. 
$\H$ then selects a random bit $b\in\{0,1\}$.

\item[\challenge]
For each round $i=1,2,\ldots$, $\A$ chooses two adaptive query messages $x_{i,0}$ and $x_{i,1}$ and sends the query messages to $\H$. 
$\H$ selects random coins $R_{i,0}, R_{i,1}\in\{0,1\}^{\lambda}$ uniformly at random, samples $\lp_{i,0} \gets \LP(\mhfeval(x_{i,0}; R_{i,0})) $ and $ \lp_{i,1} \gets \LP(\mhfeval(x_{i,1}; R_{i,1}))$, and sends the ordered pair $(\lp_{i,b},\lp_{i,1-b})$ to $\A$. Note: $\A$ can pick $x_{i+1,0}$ and $x_{i+1,1}$ adaptively after observing the response $(\lp_{i,b},\lp_{i,1-b})$. 

\item[\guess]
The game ends when the adversary $\A$ outputs $b'$ as a guess for $b$. 
$\A$ wins the game if $b=b'$.
\end{description}
\end{mdframed}

As before, the advantage of the adversary to win the adaptive data independency game for ciMHF is:
\begin{align*}
Adv_{\A, \MHF}^{\mathsf{ind-mult-lp-iMHF}} = \left|\frac{1}{2}-\PPr{\mathcal{A}(\mathcal{T})=b': b=b')}\right|,
\end{align*}
where $\mathcal{T}$ is the transcript $\{x_{i,0},x_{i,1},\lp_{i,b},\lp_{i,1-b}\}$ and $\lp_{i,j}=\LP(x_{i,j}; \mhfeval(x_{i,j}; R_{i,j}))$.

\begin{definition}
We say an evaluation algorithm $\mhfeval$ has $(t,\eps)$-single security if any attacker running in time $t$ has at most advantage $\eps$ in the data independency game. 
Similarly, we say an evaluation algorithm $\mhfeval$ has $(t,r,\eps)$-adaptive security if any attacker running in time $t$ and making $r$ queries has at most advantage $\eps$ in the adaptive data independency game. 
\end{definition}
We conclude by noting the following relationship between single security and adaptive security, thus implying the security of our evaluation function in \figref{fig:implementation} of the dMHF in \figref{fig:graph:permutation} with respect to the adaptive data independency game. 
\begin{theorem}
$(t,\eps)$-single security implies $(t-\O{r\cdot\time(\mhfeval)}, r, r\eps)$-adaptive security.
\end{theorem}
\begin{proof}
Suppose that $\aadapt$ violates $(t-\O{r\cdot\time(\mhfeval)}, r, r\eps)$-adaptive security for the sake of contradiction. 
Without loss of generality we will assume that $\aadapt$ outputs $b'=b$ with probability greater than $\frac{1}{2} + r \eps$. 
We will use $\aadapt$ to construct an attacker $\asingle$ that violates  $(t,\eps)$-single security. 

We first define a sequence of $r$ hybrids in the adaptive data-independency game. 
In Hybrid $i$, the challenger $\H$ picks bits $b, b_1,\ldots,b_{i-1}$ uniformly at random and sets $b_{i}=b,b_{i+1}=b,\ldots,b_r=b$. In round $j$ when the attacker $\aadapt$ submits two strings $x_{j,0}$ and $x_{j,1}$, the challenger $\H$ samples $\lp_{j,0} \gets \LP(\mhfeval(x_{j,0}; R_{j,0})) $ and $ \lp_{j,1} \gets \LP(\mhfeval(x_{j,1}; R_{j,1}))$ and then responds with $\lp_{j,b_i},\lp_{j,1-b_i}$ instead of $\lp_{j,b}$ and $\lp_{j,1-b}$ i.e., the bit $b_i$ is used to permute the order of the responses in round $i$ instead of $b$. 

Observe that in Hybrid 1, $b_1=\ldots=b_r=b$, so that Hybrid $1$ is equivalent to the actual adaptive independency game. 
Similarly, in Hybrid $r$, the bits $b_1,\ldots,b_r$ are all picked independently so that $\aadapt$ working in Hybrid $r$ has no advantage i.e., the attacker guesses $b'=b$ correctly with probability at most $\PPr{b'=b~|~\mbox{Hybrid }r}=\frac{1}{2}$. 
We observe that the advantage of the attacker is 
\[\PPr{b'=b~|~\mbox{Hybrid 1}}-\frac{1}{2} = \PPr{b'=b~|~\mbox{Hybrid }1} - \PPr{b'=b~|~\mbox{Hybrid }r} = \Delta_2 + \ldots + \Delta_r \ , \]
where $\Delta_i = \PPr{b'=b~|~\mbox{Hybrid }i-1}-\PPr{b'=b~|~\mbox{Hybrid }i}$. 
By an averaging argument, we must have $\Delta_{i+1} > \epsilon$ for some $i<r$. 
The following observation will also be useful:
\[ \Delta_{i+1} = \frac{1}{2} \PPr{b'=b~|~\mbox{Hybrid }i} - \frac{1}{2} \PPr{b=b'~|~\mbox{Hybrid }i+1, b_i \neq b}.\]

\paragraph{Reduction.}
We now define $\asingle$ as follows: (1) $\asingle$ simulates $\aadapt$ along with the adaptive challenger $\H_{adaptive}$. $\asingle$ generates random bits $b_1,\ldots,b_{i-1}$ and sets $b_{i+1}=\ldots = b_r=b''$ for another random bit $b''$. 
In each round  $j \neq i$, when $\aadapt$ outputs a query $x_{j,0},x_{j,1}$, our attacker $\asingle$ simply computes $\lp_{i,0} \gets \LP(\mhfeval(x_{i,0}; R_{i,0})) $ and $ \lp_{i,1} \gets \LP(\mhfeval(x_{i,1}; R_{i,1}))$ and responds with $\lp_{i,b_i},\lp_{i,1-b_i}$. 
When $\aadapt$ outputs the query $x_{i,0},x_{i,1}$  in round $i$, $\asingle$ forwards this query to the challenger for the single stage challenger $\H_{single}$ and receives back $\lp_{i,b},\lp_{i,1-b}$ for an unknown bit $b$ selected by $\H_{single}$. 
Finally, when $\aadapt$ outputs a guess $b'$ (for $b''$) $\asingle$ outputs the same guess $b'$ (for $b$). 

\paragraph{Analysis.}
Notice that since $b''$ is just a bit selected uniformly at random and independent from $b$, then $\PPr{b'=b'' ~| b''=b} = \PPr{b'=b'' ~|~\mbox{Hybrid }i}$. 
Then from the above observation, we have $\PPr{b'=b'' ~| b''\neq b} = \PPr{b'=b''~|~\mbox{Hybrid }i+1, b_i \neq b''} =\PPr{b'=b''|~\mbox{Hybrid }i}- 2\Delta_{i+1}$. 
It follows that $\PPr{b'=b'' ~|~ b''=b} -\PPr{b'=b'' ~| b''\neq b} = 2\Delta_{i+1}$. 
Thus, the probability that $\asingle$ wins is 
\begin{align*}
\PPr{b'=b} &= \PPr{b''=b}\PPr{b'=b'' ~|~ b''=b} + \PPr{b'' \neq b} (1-\PPr{b'=b'' ~| b''\neq b})\\
&= \frac{1}{2}+\Delta_{i+1} > \epsilon.
\end{align*}
Furthermore, the running time of $\asingle$ is at most $t$.
This contradicts the assumption that the evaluation algorithm $\mhfeval$ has $(t,\eps)$-single security. 
Therefore, $(t,\eps)$-single security implies $(t-\O{r\cdot\time(\mhfeval)}, r, r\eps)$-adaptive security.
\end{proof}


\section*{Acknowledgements}
We would like to thank anonymous ITCS 2020 reviewers for thorough feedback on this paper. 
This research was supported in part by the National Science Foundation under award  \#1755708 and by IARPA under the HECTOR program. 
This work was done in part while Samson Zhou was a postdoctoral fellow at Indiana University.

\def\shortbib{0}
\bibliographystyle{alpha}
\bibliography{references}
\end{document}